\documentclass[11pt]{elsarticle}
\usepackage{geometry}                % See geometry.pdf to learn the layout options. There are lots.
\usepackage{graphicx}
\usepackage{amsmath}
\usepackage{amssymb}
\usepackage{epstopdf}
\usepackage{color}
\usepackage{bbm}

\usepackage{amsthm}
\usepackage{enumitem}

\newtheorem{lemma}{Lemma}
\newtheorem{cor}{Corollary}

\usepackage[tableposition=top]{caption}

\newcommand{\bx}{\mathbf{x}}
\newcommand{\bv}{\mathbf{v}}
\newcommand{\bB}{\mathbf{B}}
\newcommand{\bb}{\mathbf{b}}
\newcommand{\bE}{\mathbf{E}}
\newcommand{\bF}{\mathbf{F}}
\newcommand{\bG}{\mathbf{G}}
\newcommand{\bu}{\mathbf{u}}
\newcommand{\bj}{\mathbf{j}}
\newcommand{\bI}{\mathbf{I}}
\newcommand{\Dt}{\Delta t}
\newcommand{\dt}{\delta t}
\newcommand{\nph}{^{n+1/2}}
\newcommand{\Oc}{\Omega_c}

\title{Asymptotic-preserving gyrokinetic implicit particle-orbit integrator for arbitrary electromagnetic fields\tnoteref{t1}}
\tnotetext[t1]{This work was performed under the auspices of the U.S.\ Department of Energy by LLNL and LANL under contracts DE-AC52-07NA27344, DE-AC52-06NA25396, and supported by the Exascale Computing Project (LFR) (17-SC-20-SC), a collaborative effort of the U.S. Department of Energy Office of Science and the National Nuclear Security Administration, and the U.S. DOE Office for Advanced Scientific Computing Research (LC).}
\author[1]{L.F.\ Ricketson \corref{cor1}}
\ead{ricketson1@llnl.gov}
\author[2]{L.\ Chac\'{o}n}
\ead{chacon@lanl.gov}

\cortext[cor1]{Corresponding author}
\address[1]{Lawrence Livermore National Laboratory, Livermore, CA 94550, United States}
\address[2]{Los Alamos National Laboratory, Los Alamos, NM 87545, United States}
\begin{document}

\begin{abstract}
We extend the asymptotic preserving and energy conserving time integrator for charged-particle motion developed in [Ricketson \& Chac\'{o}n, JCP, 2020] to include finite Larmor-radius (FLR) effects in the presence of electric-field length-scales comparable to the particle gyro-radius (the gyro-kinetic limit).  
%These improvements maintain accuracy of the orbit integration while stepping over the gyration time-scale, and relax previously derived time-step constraints to enable additional speed-ups over traditional integrators. 
We introduce two modifications to the earlier scheme.  The first is the explicit gyro-averaging of the electric field at the half time-step, along with an analogous modification to the current deposition, which we show preserves total energy conservation in implicit PIC schemes.  The number of gyrophase samples is chosen adaptively, ensuring proper averaging for large timesteps, and the recovery of full-orbit dynamics in the small time-step limit.  The second modification is an alternating large and small time-step strategy that ensures the particle trajectory samples gyrophases evenly.  We show that this strategy relaxes the time-step restrictions on the scheme, allowing even larger speed-ups than previously achievable.   We demonstrate the new method with several single-particle motion tests in a variety of electromagnetic field configurations featuring gyro-scale variation in the electric field.  The results demonstrate the advertised ability to capture FLR effects accurately even when significantly stepping over the gyration time-scale.  
\end{abstract}
%%%%%%%%%%%%%%%%%%%%%%%%%%%%%%%%%%%%%%%%%%%%%%%%%%
\maketitle

%%%%%%%%%%%%%%%%%%%%%%%%%%%%%%%%%%%%%%%%%%%%%%%%%%

\section{Introduction}
The gyrokinetic approximation of the Vlasov equation (GK) \cite{abel2013multiscale, brizard2007foundations, dubin1983nonlinear, hahm1988nonlinear, lee1983gyrokinetic} has enabled tremendous advances in both theoretical and computational understanding of strongly magnetized plasmas.  A key advantage of the gyrokinetic equations is that they allow stepping over the gyration time-scale,  $\Oc^{-1} = (qB/m)^{-1}$, even when length scales of the order of the gyroradius are present in the electromagnetic fields (leading to so-called finite-Larmor-radius effects, FLR). This, in turn, affords enormous computational speed-ups of the GK model vs. full-orbit Vlasov models.  One key application is magnetic confinement fusion, where gyrokinetic simulations remain the state of the art.  

However, the gyrokinetic approach is not without drawbacks.  Realistic boundary conditions for the gyrokinetic equations are difficult to formulate and implement due to the non-zero spatial extent of the simulated particles, which are effectively rotating rings of charge.  There is also increasing concern that the gyrokinetic approximation is insufficiently accurate in portions of the domain for problems of interest.  Examples include field-reversed configurations (FRCs) and the scrape-off layer (SOL) in certain tokamaks \cite{Chen_2009, De_2019, genoni2010fast, Hu_2018, pandit2018conversion, sturdevant2016fully, Sturdevant_2017, STURDEVANT2016519, wang:2020aa}.

As a result of these limitations, there is a growing interest in enabling the simulation of full-orbit kinetic systems when required, at least for ions \cite{sturdevant2016fully}, without losing the computational advantage of GK models when appropriate.  This can be accomplished by developing so-called asymptotic-preserving (AP) algorithms that can track the full particle orbit when the time-step $\Dt$ is small compared to the gyroperiod, $\Oc \Dt \lesssim 1$, and revert to the gyrokinetic treatment when $\Oc \Dt \gg 1$.  The desire to step over a physical time-scale clearly motivates looking at implicit time integration. We are also motivated to produce a scheme that can eventually be coupled with strictly charge- and energy-conserving implicit particle-in-cell (PIC) schemes \cite{chacon2013charge, chen2014energy, chen2015multi, chen2011energy, chen2012efficient, chen2014fluid, taitano2013development} without loss of conservation properties.

%Happily, there has also been considerable recent progress in implicit particle-in-cell (PIC) schemes \cite{chacon2013charge, chen2014energy, chen2015multi, chen2011energy, chen2012efficient, chen2014fluid, taitano2013development}.  These schemes feature exact total energy conservation for any time-step and cell size.  As a result of this, they have been shown to be more resistant to the finite-grid instability than their explicit counterparts \cite{barnes2021finite}.  

Recently, the authors introduced \cite{ricketson2019energy} an implicit, strictly energy conserving,  AP particle-push algorithm able to recover guiding-center {\em drift}-kinetic motion (i.e., correct gyroradius $\rho$ and all first-order drifts) even when $\Omega_c \Delta t \gg 1$. 
The scheme 
%represents a viable alternative to \textit{drift}-kinetic PIC simulation and has distinct advantages.  Indeed, it 
overcomes many of the obstacles mentioned above; namely, the simulated particles can deal with realistic boundary conditions as in the full-orbit kinetic model, and it enables a unified treatment of magnetized and unmagnetized regions in a single scheme. The scheme has been recently reformulated for optimal nonlinear convergence performance \cite{koshkarov2022fast} when inverting the associated nonlinear algebraic system stemming from the implicit discretization.  However, the scheme tacitly assumes that the gyroradius is much smaller than any other scales of interest (because fields are not gyroaveraged), and therefore it is not expected to capture the finite Larmor radius (FLR) effects of the \textit{gyro}-kinetic limit.  In fact, our numerical results here will confirm that it does not.  

There have been several recent studies {\cite{,filbet2016asymptotically,filbet2017asymptotically,chartier2019uniformly,chartier2020uniformly,hairer2022large,filbet2023asymptotically} attempting to develop AP particle-orbit integration schemes similar in spirit to that in Ref. \cite{ricketson2019energy}, with varying degrees of generality and complexity. However, none of these schemes conserve energy exactly (and therefore are not suitable for coupling with modern strictly energy-conserving implicit PIC codes), and to our knowledge none of them are able to deal with FLR effects to reproduce the gyrokinetic limit.

The purpose of the present work is to generalize the AP energy-conserving scheme in Ref. \cite{ricketson2019energy} to capture the necessary FLR effects in the gyrokinetic limit.  In particular, we will show that our generalization can handle an electric field $\bE$ with spatial scales comparable to the gyroradius $\rho$, i.e., $k_\perp \rho \sim 1$, with $k_\perp$ a measure of the characteristic inverse electric-field length scale.  
Our basic strategy for capturing FLR effects is relatively standard: rather than evaluating the electric field at the instantaneous particle location, we compute a gyro-averaged field using some number $n_g$ of equispaced samples along a gyro-orbit estimated from the instantaneous particle state.  The subtlety arises from two factors: our desire to maintain exact total energy conservation, and the need to preserve second-order convergence to the full-orbit motion in the limit $\Omega_c \Delta t \rightarrow 0$.  To the first point, we show that an analogous and physically reasonable modification to current deposition preserves exact energy conservation in implicit PIC schemes.  To the second, we describe a method of adaptively choosing $n_g$ based on local field conditions such that it tends to one in the small time-step limit, and we show that this single gyrophase sample can be chosen in such a way that second-order accuracy is preserved.  

In the development of this FLR capturing scheme, we have also discovered a technique for relaxing previously derived time-step constraints on the scheme \cite{ricketson2019energy}.  Recall from the reference that, in certain regimes, the scheme's time-step is constrained by 
\begin{equation}
	\Omega_c \Delta t \lesssim \min \left\{ \left( \rho \frac{\| (\nabla B)_\perp \|}{B} \right)^{-1/2}, \sqrt{\Omega_c \tau_{res}} \right\},
\end{equation}
where $\tau_{res}$ is the smallest time-scale one wishes to resolve in the simulation.  While this does permit stepping over the gyration time-scale, the presence of square roots means the scheme still requires smaller time-steps than a truly gyrokinetic particle push.  We show here that, by alternating a large ($\Omega_c \Delta t \gg 1$) time-step followed by a smaller one ($\Omega_c \delta t \sim 1$), we can relax these time-step constraints to:
\begin{equation}
	\Omega_c \Delta t \lesssim \min \left\{ \left( \rho \frac{\| (\nabla B)_\perp \|}{B} \right)^{-1}, \Omega_c \tau_{res} \right\}.
\end{equation}
In strongly magnetized plasmas, both arguments in the minimum function are much larger than one, so removing the square roots represents a considerable potential speedup that is well worth the approximate factor of two penalty paid for taking a smaller time-step at every other step.

We selectively compare the new scheme (as a test-particle scheme, not coupled with PIC at this point) against the Boris push, the present scheme's precursor \cite{ricketson2019energy}, the scheme of Brackbill, Forslund and Vu (BFV) \cite{brackbill1985simulation, vu1995accurate}, and the standard Crank-Nicolson scheme in a variety of field configurations.  We find that the modifications do enable accurate prediction of single-particle orbits also when $k_\perp \rho \sim 1$, making it the first full-orbit AP scheme to do so while stepping over the gyration time-scale.  We also show that the alternating time-step approach enables accurate results at larger time-steps than previously achievable.  

The remainder of the article is structured as follows.  In Section 2, we review the previous AP scheme including its motivation, derivation, time-step restrictions, and relationship to implicit PIC.  In Section 3, we motivate and present the new scheme.  We show how to modify implicit PIC methods to be compatible with both the new scheme and energy conservation, and we derive new, relaxed time-step restrictions on the scheme and an adaptive time-stepping strategy that respects those restrictions.  We present numerical results in Section 4 and conclude in Section 5.  

%%%%%%%%%%%%%%%%%%%%%%%%%%%%%%%%%%%%%%%%%%%%%%%%%%%%%%%%%%%
%%%%%%%%%%%%%%%%%%%%%%%%%%%%%%%%%%%%%%%%%%%%%%%%%%%%%%%%%%%
%%%%%%%%%%%%%%%%%%%%%%%%%%%%%%%%%%%%%%%%%%%%%%%%%%%%%%%%%%%

\section{Background}

\subsection{Prior energy-conserving and asymptotic-preserving time integrator}
The authors' previous energy-conserving and asymptotic-preserving (AP) scheme \cite{ricketson2019energy} is built upon Crank-Nicolson while borrowing ideas from \cite{brackbill1985simulation, genoni2010fast, vu1995accurate}.  The scheme is defined by (using the reformulation of the effective force proposed in Ref. \cite{koshkarov2022fast}):
\begin{equation} \label{eq:CBFVdef}
\begin{split}
	\bx^{n+1} &= \bx^n + \Dt \bv^{n+1/2}, \\
	\bv^{n+1} &= \bv^n + \Dt \frac{q}{m} \left[ \bE^{n+1/2} + \bv^{n+1/2} \times \bB^{n+1/2} + \mathbf{F}_{eff} \right], \\
	%\mathbf{F}_{eff} &= \left[ \bb - \frac{v_\parallel}{v_\perp} \frac{\bv_\perp}{v_\perp} \right] G_\parallel\nph + \left[ \bI - \frac{\bv_\perp \bv_\perp}{v_\perp^2} \right] \cdot \mathbf{G}_\perp\nph, \\
    \mathbf{F}_{eff} &= \mathbf{v}\times\left(\frac{\mathbf{G}^{n+1/2}\times\mathbf{v}_{\perp}}{v_{\perp}^{2}}\right) \\
	\mathbf{G}\nph &= F_{B,\parallel} \bb + \frac{(\bI - \widehat{\bv}_E \widehat{\bv}_E) \cdot \bF_{B,\perp}}{1 - \eta^2/2} + \frac{2}{\eta^2} \widehat{\bv}_E \left( \widehat{\bv}_E \cdot \bF_{B,\perp} + \mathbbm{1}_{v_E > u} F_{B,\parallel} \frac{v_\parallel}{v_\perp} \right), \\
	\bF_{B} &= -\tilde{\mu} \nabla B^{n+1/2}, \qquad \tilde{\mu} = m \frac{ \left\| \bv_\perp^{n+1} - \bv_\perp^n \right\|^2}{8B^{n+1/2}}.
\end{split}
\end{equation}
In the above, quantities at the half-step are defined by
\begin{equation}
\begin{split}
	\bv\nph &= \left( \bv^n + \bv^{n+1} \right)/2, \quad \bx\nph = \left( \bx^n + \bx^{n+1} \right) / 2, \\
	\bE\nph &= \bE \left( \bx\nph, t\nph \right), \quad \bB_p\nph = \bB \left( \bx\nph, t\nph \right).
\end{split}
\end{equation}
Moreover, $\bv_E = \bv_E\nph = \bE\nph \times \bB\nph / (B\nph)^2$, $\widehat{\bv}_E = \bv_E / v_E$, $\bb = \bb\nph = \bB\nph / B\nph$, $\bu = \bu\nph = \bv_\perp\nph - \bv_E$, and $\eta = \min \left\{ 1, u/v_E \right\}$.  All velocities lacking a superscript are evaluated at time $t\nph$.  The subscripts $\perp$ and $\parallel$ refer to vector components orthogonal and parallel to $\bb$, respectively.  Finally, $\mathbbm{1}$ denotes the indicator function - that is, $\mathbbm{1}_A$ is one if $A$ is true and zero otherwise.  

The full derivation of the scheme is found in \cite{ricketson2019energy}, but we briefly summarize the motivation and overarching approach here.  If one were to set $\bF_{eff} = 0$, the first two lines above are just Crank-Nicolson (CN).  Without modification,  CN achieves the correct gyroradius for arbitrary $\Dt$ (the reasoning is the same as that used to conclude that the Boris integrator features an artifically large gyroradius for $\Oc \Dt \gg 1$, see e.g.\ \cite{birdsall2004plasma, parker1991numerical}).  Moreover, the analysis of Boris in \cite{parker1991numerical} can be straightforwardly generalized to Crank-Nicolson to show that it also captures all first-order guiding center drifts \textit{except} for the $\nabla B$ drift and mirror force.  

To extend Crank-Nicolson to capture the missing $\nabla B$ force, Brackbill, Forslund and Vu (henceforth BFV) \cite{brackbill1985simulation, vu1995accurate} proposed introducing the effective force $\bF_B$.  This force has the properties:
\begin{equation} \label{eq:BFV_force_limits}
\begin{split}
	\bF_B \approx -\mu \nabla B &\quad \text{for} \quad \Oc \Dt \gg 1, \\
	\bF_B = O\left(\Dt^2\right) &\quad \text{for} \quad \Oc \Dt \ll 1.
\end{split}
\end{equation}
Here, $\mu = mu^2 / 2B$ is the particle's magnetic moment.  As a result, this force is effective in capturing the missing drift motions when $\Oc \Dt \gg 1$ and preserves the second-order convergence of the scheme to full orbit dynamics as $\Dt \rightarrow 0$.  The drawback is that $\bF_B$ can do work on the particle, since there is no guarantee that $\bv\nph \cdot \bF_B = 0$.  This can lead to $O(1)$ changes in energy that have significant deleterious effects on the accuracy of the computed orbits; see the numerical examples in \cite{ricketson2019energy}.  

The scheme \eqref{eq:CBFVdef} improves upon BFV by enforcing energy conservation, $\bv\nph \cdot \bF_{eff} = 0$, in addition to capturing all first-order drifts for $\Oc \Dt \gg 1$.  Energy conservation is achieved for arbitrary $\bG$ by the third line of \eqref{eq:CBFVdef}.  It can then be shown that the form of $\bG$ given is the unique expression satisfying both
\begin{equation} \label{eq:FeffConstraints}
\begin{split}
	\left\langle \bF_{eff} \right\rangle &= \bF_B, \\
	\left\langle \bG \right\rangle &= \bG,
\end{split}
\end{equation}
where $\langle \cdot \rangle$ denotes the gyro-average operator.  The first line above ensures the AP properties we desire, while the second merely makes the algebraic derivation tractable.  The expression for $\bG$ is valid under the assumption that the particle motion is dominated jointly by gyration and the $\bE \times \bB$ drift.  It is assumed that if this is not the case, the gyro-period must be resolved, meaning that $\bF_{eff} = O\left(\Dt^2\right)$ is negligible in any case.  

A key observation that will be leveraged in our subsequent analysis is that $\bF_{eff}$ cannot be made to approximate $-\mu \nabla B$ pointwise in time while simultaneously conserving energy.  This is clear in the case when gyration dominates the particle motion: $\bv\nph$ rotates rapidly at the gyrofrequency, while $\mu \nabla B$ varies over much slower time-scales, so one cannot enforce $\bv\nph \cdot \mu \nabla B = 0$ at all times.  This is what forces us to settle in \eqref{eq:FeffConstraints} for imposing constraints on the \textit{gyro-average} of $\bF_{eff}$ rather than $\bF_{eff}$ itself.  The scheme relies on a sequence of consecutive time-steps to sample the particle at various gyrophases to accumulate an approximate gyro-average of $\bF_{eff}$.  This idea motivates the use of alternating large and small time-steps presented below, as this can force the particle to traverse a prescribed gyrophase at each time-step.  

Clearly, the reliance on time-stepping to compute gyro-averages introduces a new numerical time-scale that must be resolved in addition to physical scales.  Namely, the time it takes to accumulate an accurate gyro-average should be small compared to time-scales of interest.  Additionally, on shorter time-scales the particle may drift due to the fact that $\bF_{eff}$ deviates from its desired gyro-average.  One should ensure that the distance the particle drifts due to this unphysical effect is small compared to length-scales of interest.  These constraints were analyzed quantitatively in \cite{ricketson2019energy} to arrive at upper bounds on $\Oc \Dt$ based on physical parameters in two asymptotic limits: $u\nph \gg v_E$ and $u\nph \ll v_E$.  These results are summarized in Table \ref{table:dt_table}.  Therein, $\tau_\textrm{res}$ is a shorthand for the fastest time-scale one wishes to resolve, and 
\begin{equation} \label{eq:deltadefs}
	\delta_\perp = \rho \frac{\left\| (\nabla B)_\perp \right\|}{B}, \qquad \delta_\parallel = \frac{v_\parallel}{\Oc} \frac{\left\| (\nabla B)_\parallel \right\|}{B}, \qquad \delta_E = \frac{v_E}{\Oc} \frac{\left\| (\nabla B)_\perp \right\|}{B}.
\end{equation}
These three non-dimensional parameters have simple interpretations: $\delta_\perp$ is the fractional change in $B$ over a gyro-orbit; $\delta_\parallel$ is the fractional change in $B$ over a gyration time due to parallel motion; $\delta_E$ is the fractional change in $B$ over a gyration time due to $\bE \times \bB$ motion.  As before, the expressions in Table \ref{table:dt_table} are valid under the assumption that the particle's velocity is jointly dominated by gyration and $\bE \times \bB$ motion.  

\begin{table}[t]
\centering
\vspace{1em}
\caption{The maximum allowable value of $\Oc \Dt$ to ensure (a) the anomalous displacement does not exceed gyroradius and (b) the time-scale over which accurate implicit gyroaveraging is achieved is shorter than $\tau_{\textrm{res}}$.  Expressions are valid in the asymptotic limits (i) $u\nph \gg v_E$ and (ii) $u\nph \ll v_E$, respectively.} \label{table:dt_table}
\begin{tabular}{c | c c}
 $\max \left( \Oc \Dt \right)$ & (a) Displacement restriction & (b) Avg. restriction \\
 \hline
 (i) $u^{n+1/2} \gg v_E$ & $2 \min \left\{ \sqrt{2} (\delta_\perp)^{-1/2},  (\delta_\parallel)^{-1/2} \right\}$ & $2 \left( \frac{\Omega_c \tau_{\textrm{res}}}{\pi} \right)^{1/2}$ \\
 (ii) $u^{n+1/2} \ll v_E$ & $\sqrt{2} \left( \delta_E + \delta_\parallel \right)^{-1/2}$ & $\Omega_c \tau_{\textrm{res}}$
\end{tabular}
\vspace{1em}
\end{table}

\subsection{Implicit PIC: Energy and local charge conservation}
The scheme described above was designed to function as a piece of the implicit PIC schemes developed over the last decade \cite{chacon2013charge, chen2014energy, chen2015multi, chen2011energy, chen2012efficient, chen2014fluid, ricketson2023pseudo, taitano2013development}.  The essence of these schemes can be stated compactly in the simple collisionless, electrostatic case:
\begin{equation} \label{eq:implicitPIC}
\begin{split}
	\bx^{n+1}_p &= \bx^n_p + \bv_p\nph, \\
	\bv^{n+1}_p &= \bv^n_p + \frac{q_p}{m_p} \left[ \bE\nph_p + \bv\nph_p \times \bB\nph_p \right], \\
	\bE_p\nph &= \sum_i \bE\nph_i S \left( \bx_p\nph - \bx_i \right), \\
	\bE_i^{n} &= -\nabla_h \phi_i^n, \\
	\nabla_h^2 \phi_i^{n+1} &= \nabla_h^2 \phi_i^n + \frac{\Dt}{\epsilon_0} \nabla_h \cdot \bj^{n+1/2}_i, \\
	\bj_i^{n+1/2} &= \frac{1}{ \lvert \mathbf{h} \rvert } \sum_p \bv_p\nph S \left( \bx_i - \bx_p\nph \right).
\end{split}
\end{equation}
Here, the subscript $p$ indicates a particle quantity while $i$ indicates a grid quantity.  $S$ is the so-called ``shape function" that interpolates between the grid and particles, and $\nabla_h$, $\nabla_h \cdot$, $\nabla^2_h$ are discrete gradient, divergence, and Laplacian operators, respectively.  $\bB$ is assumed to be a known function that may be evaluated at particle locations at will, $\lvert \mathbf{h} \rvert$ denotes the cell volume, and $\epsilon_0$ is the permittivity of free space.  Note that the first two lines are simply the Crank-Nicolson time integrator for particles, while the fifth line is a discretization of the divergence of Ampere's equation in the electrostatic limit.  

This implicit discretization of the Vlasov-Ampere system has been carefully constructed to ensure exact energy conservation for any spatial resolution $\mathbf{h}$ and time-step $\Delta t$.  Stated precisely, one has:
\begin{equation} \label{eq:EconsStatement}
	\left[ \sum_p \frac{1}{2} m_p v_p^2 + \frac{\epsilon_0}{2} \lvert \mathbf{h} \rvert \sum_i E_i^2 \right]^{n+1} = \left[ \sum_p \frac{1}{2} m_p v_p^2 + \frac{\epsilon_0}{2} \lvert \mathbf{h} \rvert \sum_i E_i^2 \right]^{n}.
\end{equation}
The proof of this fact in several different contexts is well-described in the citations above, but we re-print it here because following its logic instructs later modifications arising from our new time integrator. The proof reads: 
\begin{equation} \label{eq:EconsProof}
\begin{split}
	\frac{1}{2} \sum_p m_p \left( \left\| \bv_p^{n+1} \right\|^2 - \left\| \bv_p^n \right\|^2 \right) &= \sum_p m_p \bv_p^{n+1/2} \cdot \left( \bv_p^{n+1} - \bv_p^n \right) \\
	\textrm{CN definition} \rightarrow \quad &= \Delta t \sum_p q_p \bv_p^{n+1/2} \cdot \left( \bE\nph_p + \bv_p^{n+1/2} \times \bB_p\nph \right) \\
	\textrm{cross prod. + def. of } \bE_p \rightarrow \quad &= \Delta t \sum_p \sum_i q_p \bv_p^{n+1/2} \cdot \bE^{n+1/2}_i S \left(\bx_p^{n+1/2} - \bx_i \right) \\
	\textrm{swap sums + } \bj \textrm{ def.} \rightarrow \quad &= \Delta t | \mathbf{h} | \sum_i \bE^{n+1/2}_i \cdot \bj_i^{n+1/2} \\
	\bE = -\nabla_h \phi \rightarrow \quad &= -\Delta t | \mathbf{h} | \sum_i \bj^{n+1/2}_i \cdot \nabla_h \phi_i^{n+1/2} \\
	\textrm{integrate by parts} \rightarrow \quad &= \Delta t | \mathbf{h} | \sum_i \left( \nabla_h \cdot \bj^{n+1/2}_i \right) \phi_i^{n+1/2} \\
	\textrm{def. of Ampere update} \rightarrow \quad &= \epsilon_0 | \mathbf{h} | \sum_i \left( \nabla^2_h \phi_i^{n+1} - \nabla^2_h \phi_i^n \right) \phi_i^{n+1/2} \\
	\textrm{integrate by parts} \rightarrow \quad &= -\epsilon_0 | \mathbf{h} | \sum_i \left( \nabla_h \phi_i^{n+1} - \nabla_h \phi_i^n \right) \cdot \nabla_h \phi_i^{n+1/2} \\
	&= -\frac{\epsilon_0}{2} | \mathbf{h} | \sum_i \left\{ \left\| \bE_i^{n+1} \right\|^2 - \left\| \bE_i^n \right\|^2 \right\}.
\end{split}
\end{equation}
Simple rearrangement of the first and last expressions yields \eqref{eq:EconsStatement}.  Note that we have assumed the discrete differentiation operators preserve the integration-by-parts properties enjoyed by their continuum analogues.  In practice, some care is needed in choosing $\nabla_h$, $\nabla_h \cdot$, and $\nabla^2_h$ so that this is true.  However, as the subject of this article is the particle advance, we leave those details to the references above.  

Note that in moving from the second line to the third in \eqref{eq:EconsProof}, any effective force $\bF_{eff}$ that is added to Crank-Nicolson satisfying $\bv_p\nph \cdot \bF_{eff} = 0$ vanishes.  Thus, the previously developed AP time-integrator \eqref{eq:CBFVdef} may be substituted into implicit PIC for Crank-Nicolson while preserving discrete energy conservation.  Similarly, we will show that the modifications introduced here will preserve this desirable property when benign changes to relevant definitions are made.  

A second important conservation property that should be respected is local charge conservation.  That is, a discrete continuity equation should be satisfied.  With a second-order, centered finite difference spatial discretization, this takes the form
\begin{equation} \label{eq:discrete_continuity}
	\rho_{i+1/2}^{n+1} = \rho_{i+1/2}^n - \Dt \nabla_h \cdot \bj_i\nph,
\end{equation}
where $\rho$ is a charge density.  The combination of Ampere's law and this continuity equation enforces Gauss' law.  Thus, enforcement of continuity allows the discrete Vlasov-Ampere system to avoid the need for so-called ``cleaning" of the field at each time-step to enforce Gauss' law.  Instead, Gauss's law is automatically satisfied at every time-step.  

There are multiple methods in the literature for enforcing this discrete continuity equation \cite{chen2011energy, chen2019semi, esirkepov2001exact}, but the vast majority are based on the observation that \textit{within a single cell} it is often possible to find a shape function $\tilde{S}$ (typically different from the $S$ appearing above, and in fact it may be necessary for $S$ to be different for different components of $\mathbf{j}$) such that:
\begin{equation} \label{eq:charge_detailed_balance}
	\tilde{S}(\mathbf{x}_{i+1/2} - \mathbf{x}_p^{n+1}) - \tilde{S}(\mathbf{x}_{i+1/2} - \mathbf{x}_p^n) + \nabla_h \cdot \left[ (\mathbf{x}_p^{n+1} - \mathbf{x}_p^n) S(\mathbf{x}_i - \mathbf{x}_p^{n+1/2}) \right] = 0.
\end{equation}
This identity often requires leveraging the fact that a centered finite difference is exact for quadratic functions, which places constraints on the choices of $S$ and $\tilde{S}$ that can be made.  Recalling that $\mathbf{x}_p^{n+1} - \mathbf{x}_p^n = \Delta t \mathbf{v}_p^{n+1/2}$, dividing through by $| \mathbf{h} |$ and summing over all particles, one observes that this is sufficient to guarantee charge conservation with the appropriate definition of $\rho$:
\begin{equation}
	\rho^n_{i+1/2} = \frac{1}{|\mathbf{h}|} \sum_p \tilde{S}( \mathbf{x}_{i+1/2} - \mathbf{x}_p^n).
\end{equation}

Of course, one still must reckon with cell crossings, across which shape functions are non-smooth and \eqref{eq:charge_detailed_balance} is generally violated.  One typically deals with this by subdividing particle trajectories within a time-step into sub-steps that terminate on cell boundaries, and performing current deposition and field interpolation directly on these sub-steps.  See the references above for more detail, including demonstrations that the energy conservation derivation above carries through without meaningful changes due to the sub-steps.

%%%%%%%%%%%%%%%%%%%%%%%%%%%%%%%%%%%%%%%%%%%%%%%%%%%%%%%%%%%
%%%%%%%%%%%%%%%%%%%%%%%%%%%%%%%%%%%%%%%%%%%%%%%%%%%%%%%%%%%
%%%%%%%%%%%%%%%%%%%%%%%%%%%%%%%%%%%%%%%%%%%%%%%%%%%%%%%%%%%

\section{Finite Larmor Radius Effects and Time-Step Optimization}
As discussed in the preceding section, the time-integration scheme \eqref{eq:CBFVdef} accurately reproduces the particle gyroradius and all guiding center drifts up to first order for $\Oc \Dt \gg 1$.  It does this while still converging to the full-orbit dynamics at second order as $\Dt \rightarrow 0$ and conserving energy.  While this is a promising step, the guiding center drifts in question are derived under the assumption of infinitesimal gyroradius.  That is, it is assumed that $\rho / L \ll 1$, where $L$ is the smallest length scale of interest.  This is an assumption of the \textit{drift}-kinetic limit, but not of the \textit{gyro}-kinetic limit.  

To extend to the gyrokinetic limit, we must allow $\bE$ to vary on spatial scales comparable to $\rho$ in the perpendicular direction.  This spatial scale length is often written in terms of a perpendicular wave number $k_\perp$.  In this notation, the case we are concerned with is $k_\perp \rho \sim 1$.  While derivations of gyrokinetic Vlasov systems can be quite complex, the modification to single particle drift motion to handle this case is relatively simple: wherever $\bE$ appears, it should be replaced by its gyro-average $\langle \bE \rangle$.  That is, it is no longer sufficient to evaluate $\bE$ at the gyrocenter; it must be averaged along the particle's gyro-orbit.  

In previous work, we relied on a sequence of consecutive time-steps to compute a gyro-average.  In particular, we used the fact that
\begin{equation}
	\frac{1}{N}\sum_{n=1}^N \bF_{eff}\nph \approx \left\langle \bF_{eff} \right\rangle.
\end{equation}
A potential strategy, then, to evaluate $\langle \bE \rangle$ is to again rely on a sequence of consecutive time-steps.  The challenge would then be to arrange the scheme so that 
\begin{equation} \label{eq:gavgE}
	\frac{1}{N}\sum_{n=1}^N \bE \nph \approx \left\langle \bE \right\rangle.
\end{equation}
This is considerably more difficult than what was previously done with $\bF_{eff}$.  There, no field contained structure on the scale of the gyro-radius, so it sufficed to allow the time-stepping to sample an uncontrolled sequence of gyrophases.  When $\bE$ varies on scales comparable to $\rho$, the gyro-average must be performed with considerably more care if accuracy is to be expected.  Moreover, the electric field may have structure in the parallel direction, so sampling distinct gyrophases at distinct times in scenarios that feature non-trivial parallel particle motion can lead to further difficulties.  

As a result of these considerations, we instead propose to \textit{explicitly} compute a gyro-average of $\mathbf{E}$ at each time-step and to use this gyro-average to push particles rather than the pointwise value of $\mathbf{E}$.  However, care is needed to ensure that exact energy conservation is preserved, as well as convergence to full-orbit motion.  We show how this is accomplished in the following two subsections.  Afterward, we describe our strategy for alternating large and small time-steps, which results in a more accurate approximate gyroaverage of $\mathbf{F}_{eff}$ after a smaller number of time-steps $N$.  This results in relaxed time-step constraints, which we derive, in addition to showing that the numerical gyroradius is unaffected by radical changes in the sizes of consecutive time-steps.  

%To tackle this challenge, we have attempted two separate modifications to the scheme.  First, to perform a more careful gyro-average that still relies on a sequence of consecutive time-step, we employ an alternating time-step strategy to control which gyro-phases are sampled.  In particular, after each large time-step we take a smaller one that ensures a fixed gyro-angle is traversed.  This leads to equispaced samples in gyro-phase, so that we can accumulate an arbitrarily accurate gyro-average over time.  We show rigorously that this alternate stepping still yields the correct gyro-radius and (approximate) gyro-center for arbitrary $\Dt$.  While this method of gyro-averaging will ultimately prove insufficient to approximate $\langle \mathbf{E} \rangle$ in complex field configurations, we retain it because it does improve the scheme's estimation of $\langle \mathbf{F}_{eff}$.  We show that this improvement results in considerable relaxation of the time-step constraints in Table 1.  This alternating-step strategy is thus retained in the final scheme.  

%Second, we introduce an explictly-computed gyroaverage of $\mathbf{E}$ at each time-step.  We use known particle data to estimate the gyro-center and radius, and evaluate the gyro-averged field by sampling several equispaced gyrophases.  The number of samples is chosen adaptively so as to preserve convergence as $\Delta t \rightarrow 0$.  We also show how the procedure for current deposition in implicit PIC can be modified to preserve exact energy conservation.    

Armed with new, relaxed time-step constraints, we present an adaptive time-stepping procedure.  This also requires reckoning with FLR effects.  Next, we summarize the scheme and its justification in complex-field configurations.  Potential strategies for maintaining exact local charge conservation in a PIC context are discussed along with energy conservation, but their full development and implementation is left to future work.    

%%%%%%%%%%%%%%%%%%%%%%%%%%%%%%%%%%%%%%%%%%%%%%%%%%%%%%%%%%%%%%%%%%%%%%%%

\subsection{Electric Field Evaluation}
Traditionally, in implicit PIC schemes, the electric field used to push particles $\bE_p\nph$ is evaluated at the midpoint of the time-step in both space in time.  For single-particle motion with an analytically prescribed electric field, this is written mathematically as:
\begin{equation} \label{eq:midpteval}
	\bE\nph = \bE \left( \bx\nph, t\nph \right).
\end{equation}
In the context of a PIC scheme with $\bE$ defined on a spatial grid, this is written as:
\begin{equation}
    \label{eq:e-fld-scatter}
	\bE_p\nph = \sum_i \bE_i\nph S \left(\bx_i - \bx_p\nph\right),
\end{equation}
where $S$ is the particle shape function, $\bx_i$ the grid points, and $\bE_i\nph = (\bE_i^n + \bE_i^{n+1})/2$ is the half-step field at the grid points.  

The desire to capture FLR effects necessitates a change to this midpoint evaluation.  The simplest way to see this is to note that for $\Oc \Dt \gg 1$ the particle roughly traverses a diameter of the gyro-orbit during a time-step.  Thus, $\bx_p\nph$ is approximately the gyro-center position, rather than residing anywhere on the gyro-orbit.  An average over consecutive time-steps will thus not approximate the gyro-average, but rather be biased in favor of the field at the gyro-center.

We wish instead to approximate the gyroaverage of $\mathbf{E}$ at the half-step.  To do this, one first needs to use the current particle state to estimate the gyro-center and gyro-radius, $\mathbf{x}_c$ and $\rho$, respectively.  We again assume that the particle's motion is jointly dominated by gyration and the $\mathbf{E} \times \mathbf{B}$ drift, so to a good approximation:
\begin{equation}
	\mathbf{u} = \mathbf{v}_\perp - \mathbf{v}_E, \qquad \mathbf{x}_c = \mathbf{x} + \Omega_c^{-1} \mathbf{u} \times \mathbf{b}, \qquad \rho = \frac{u}{\Omega_c}
\end{equation}
at any given instant.  

As the rest of our scheme is time-centered, we choose time-centered approximations of $\mathbf{x}_c$ and $\rho$. %
Define an estimate of the gyrocenter location and gyroradius at integer times $t^n$ by
\begin{equation}
	\mathbf{x}_c^n = \mathbf{x}^n + \frac{1}{\Omega_c^n} \mathbf{u}^n \times \mathbf{b}^n, \qquad \rho^n = \frac{u^n}{\Omega_c^n}
\end{equation}
and at half-steps by
\begin{equation} \label{eq:gcdefs}
	\mathbf{x}_c^{n+1/2} = \frac{1}{2} \left( \mathbf{x}_c^n + \mathbf{x}_c^{n+1} \right), \qquad \rho^{n+1/2} = \frac{1}{2} \left( \rho^n + \rho^{n+1} \right).
\end{equation}
We then define the estimated gyro-averaged $\bE$-field as:
\begin{equation} \label{eq:gavgE}
	\left\langle \bE \right\rangle\nph = \frac{1}{n_g} \sum_{k=1}^{n_g} \bE \left( \bx_c^{n+1/2} + \rho^{n+1/2} \mathbf{n}_k, t^{n+1/2} \right),
\end{equation}
where the $\mathbf{n}_k$ are unit vectors such that $\mathbf{n}_k \cdot \bb\nph = 0$ and $\mathbf{n}_{k+1} = \mathcal{R}_{2\pi/n_g} \mathbf{n}_k$.  Here, $\mathcal{R}_\theta$ is the rotation operator about $\bb$ by angle $\theta$.  Clearly, $n_g$ is the number of gyrophase samples used.  Additionally, it is clear from Eq. \eqref{eq:e-fld-scatter} how to generalize this statement to an electric field defined on a mesh in a self-consistent PIC scheme.  

It is immediately evident that this evaluation approximates the gyro-average of the electric field.  However, it is less clear how to guarantee that this modification preserves second-order convergence to the full-orbit motion in the small time-step limit.  We achieve this in two stages.  First, we adaptively choose $n_g$ such that it tends to one as $\Delta t \rightarrow 0$.  Second, we choose the first unit vector $\mathbf{n}_1$ intelligently so that $\mathbf{x}_c\nph + \rho\nph \mathbf{n}_1 \approx \bx\nph$. 

The first step is achieved by choosing $n_g$ adaptively according to:
\begin{equation} \label{eq:adaptive_ng}
	n_g = \left\lceil \min \left\{ a (k_\perp \rho)^c, b \Omega_c \Dt \right\} \right\rceil,
\end{equation}
with $\lceil ~ \rceil$ the ceiling function, and $a$, $b$ and $c$ free parameters to be determined later. In this way, if we are either taking small time-steps or the spatial structure is on scales much larger than the gyroradius, we choose $n_g = 1$.  We use increasing values of $n_g$ as the time-step gets larger and the spatial scales become more comparable to $\rho$.  We discuss the specification of $a$, $b$, and $c$ in more detail in the adaptive time-stepping and numerical results sections, as well as the estimation of $k_\perp$ from local field data.  

It just remains to choose the proper initial unit normal vector $\mathbf{n}_1$ such that in the limit of small $\Delta t$ and $n_g = 1$, we have:
\begin{equation}
	\bx_c\nph + \rho\nph \mathbf{n}_1 = \bx\nph + O(\Delta t^2).
\end{equation}
This will ensure that the scheme is still convergent to the full-orbit solution at second order as $\Delta t \rightarrow 0$. The last equation suggests that $\mathbf{n}_1 \propto \bx_c\nph-\bx\nph$, and therefore  we choose:
\begin{equation}
	\mathbf{n}_1 = \frac{\bx^{n+1/2} - \bx_c^{n+1/2}}{\left\| \bx^{n+1/2} - \bx_c^{n+1/2} \right\|}  = -\frac{ \bu^{n+1/2} \times \bb^{n+1/2} }{ u^{n+1/2} }, 
\end{equation}
which manifestly has unit magnitude and is orthogonal to $\bb$, as required.  Moreover, elementary algebra and the definitions in \eqref{eq:gcdefs} reveal that
\begin{equation} \label{eq:gcformula}
	\bx_c^{n+1/2} + \rho^{n+1/2} \mathbf{n}_1 = \bx^{n+1/2} + \left( \bu^{n+1/2} \times \bb^{n+1/2} \right) \Omega_c^{-1} \left( 1 - \frac{u^{n+1} + u^n}{2 u^{n+1/2} }\right).
\end{equation}
Note that, when considering the small time-step limit, we may ignore the effect of $\mathbf{F}_{eff}$ when seeking second-order convergence, as it is second-order in this limit by construction -- see \eqref{eq:BFV_force_limits}-\eqref{eq:FeffConstraints}, and \cite{ricketson2019energy} for more details.  Thus, up to second order, the update for $\bu$ is simply Crank-Nicolson, which satisfies $\bu^{n+1} = \mathcal{R}_\theta \bu^n$, with the rotation angle $\theta$ given by \cite{ricketson2019energy}
\begin{equation} \label{eq:rotangle}
	\cos \theta = \frac{1 - \Oc^2 \Dt^2 / 4}{1 + \Oc^2 \Dt^2 / 4}.
\end{equation}
Thus, 
\begin{equation}
	\left( \frac{u^{n+1} + u^n}{2 u^{n+1/2}} \right)^2 = \left( u^n \right)^2 \frac{4}{\left\| R_\theta \bu^n + \bu^n \right\|^2} = \frac{2}{1 + \cos \theta} = 1 + \Omega_c^2 \Dt^2/4.
\end{equation}
As a result, 
\begin{equation}
	1 - \frac{u^{n+1} + u^n}{2 u^{n+1/2}} = 1 - \sqrt{1 + \Omega_c^2 \Dt^2 / 4} = -\Omega_c^2 \Dt^2/8 + O(\Delta t^4).
\end{equation}
This, along with \eqref{eq:gcformula}, confirms that indeed $\bx\nph$ and $\bx_c\nph + \rho\nph \mathbf{n}_1$ differ by $O(\Delta t^2)$ in the small step limit, so second-order accuracy of the overall scheme is preserved as a consequence of the second-order accuracy of Crank-Nicolson itself.  

%%%%%%%%%%%%%%%%%%%%%%%%%%%%%%%%%%%%%%%%%%%%%%%%%%%

\subsection{Energy and Charge Conservation} \label{sec:Econs}

Considerable care has been taken in developing implicit PIC methods of the type \eqref{eq:implicitPIC} to ensure that they feature exact energy and local charge conservation.  It may initially appear that the redefinition of $\bE_p\nph$ proposed here [see \eqref{eq:gavgE}] breaks these properties.  However, we show next that, with analogous changes to the current and charge deposition processes, energy and charge conservation can be recovered.  For the moment, our proof of charge conservation is limited to the case in which $n_g$ is independent of the time-step, but future work will seek to generalize this.

%To recover energy conservation, for Method 1 it is sufficient to make a change to the definition of $\bj_i\nph$ that is analogous to the change in the $\bE_p\nph$ definition.  Let
%\begin{equation} \label{eq:EconsNew}
%	\bj_i\nph = \frac{1}{| \mathbf{h} |} \sum_p \bv_p\nph \frac{ S \left( \bx_i - \bx_p^{n+1} \right) + S \left( \bx_i - \bx_p^n \right) }{2}.
%\end{equation}
%Then, following the logic of \eqref{eq:EconsProof}, we have
%\begin{equation}
%\begin{split}
%	\frac{1}{2} \sum_p m_p \left( \left\| \bv_p^{n+1} \right\|^2 - \left\| \bv_p^n \right\|^2 \right) &= \sum_p m_p \bv_p^{n+1/2} \cdot \left( \bv_p^{n+1} - \bv_p^n \right) \\
%	\textrm{CN definition} \rightarrow \quad &= \Delta t_n \sum_p q_p \bv_p^{n+1/2} \cdot \left( \bE\nph_p + \bv_p^{n+1/2} \times \bB_p\nph \right) \\
%	\textrm{cross prod. + def. of } \bE_p \rightarrow \quad &= \Delta t_n \sum_p \sum_i q_p \bv_p^{n+1/2} \cdot \bE^{n+1/2}_i \frac{ S \left(\bx_p^{n+1} - \bx_i \right) + S \left( \bx_p^n - \bx_i \right) }{2} \\
%	\textrm{swap sums + } \bj \textrm{ def.} \rightarrow \quad &= \Delta t_n | \mathbf{h} | \sum_i \bE^{n+1/2}_i \cdot \bj_i^{n+1/2}. \\
%\end{split}
%\end{equation}
%Following the subsequent steps of \eqref{eq:EconsProof} exactly shows that energy conservation is recovered with the new definitions of $\bE_p\nph$ and $\bj_i\nph$.  

To recover energy conservation, we make the following modification to the current deposition definition in \eqref{eq:implicitPIC}:
\begin{equation}
    \label{eq:gk-current}
	\mathbf{j}_i\nph = \frac{1}{|\mathbf{h}|} \sum_p \bv_p\nph \frac{1}{n_g} \sum_{k=1}^{n_g} S \left( \bx_c\nph + \rho\nph \mathbf{n}_k - \bx_i \right).
\end{equation}
Along with the definition of $\bE_p\nph$ in \eqref{eq:expgavgE}, the logic in \eqref{eq:EconsProof} carries through unchanged to recover exact energy conservation.  Note that this definition of current density makes particular physical sense in the gyrokinetic context, as it distributes current along the gyro-orbit.  Moreover, in the large time-step limit, $\bv_p\nph$ is dominated not by the particle gyration velocity [at the half-step the gyration velocity vanishes in the large step limit; see \cite{ricketson2019energy} and \eqref{eq:halfstepu}], but by drift velocities.  Thus, to leading order for large time-steps, Eq. \eqref{eq:gk-current} approximates the current density due to drift motion distributed over the gyro-orbit, as physically desired.

Turning to charge conservation, we recall that what is required is adherence to some discrete continuity equation along the lines of \eqref{eq:discrete_continuity}, and that this is typically achieved by finding a shape function satisfying \eqref{eq:charge_detailed_balance} within a cell and decomposing particle displacements into substeps that lie within a cell.  We go into more detail here on strategies for division into substeps.  In previous implicit PIC schemes, one either directly chooses time-step sizes so that particles stop on cell boundaries \cite{chen2011energy} or one chooses time-steps freely and interprets particle trajectories (for the purposes of current deposition and $\bE$ evaluation) within a time-step as composed of segments, each of which lies within a cell \cite{chen2019semi, esirkepov2001exact, umeda2003new, villasenor1992rigorous}. 

The former strategy will prove clearly incompatible with the method of alternating large and small time-steps presented below. Indeed, we will require detailed control over time-step size in order to maximize efficiency.  We will thus focus the remainder of our discussion on the latter strategy.  Note that if we have found an $\tilde{S}$ such that \eqref{eq:charge_detailed_balance} holds for any pair of points $\mathbf{x}_p^{n+1}$, $\bx_p^n$ which lie in a single cell, then the same relation holds equally well for the pair of points $\mathbf{x}_{c}^{n+1} + \rho^{n+1} \mathbf{n}_k$, $\mathbf{x}_{c}^{n} + \rho^n \mathbf{n}_k$.  Also note that, critically, the point at which we evaluate the current's shape function is the mean of these two points.  Summing the resulting version of \eqref{eq:charge_detailed_balance} from $k=1$ to $n_g$ and dividing by $n_g$ as well as $| \mathbf{h} |$, one finds that the discrete continuity equation is indeed satisfied with the definition
\begin{equation}
	\rho_{i+1/2}^n = \frac{1}{|\mathbf{h}|} \sum_p \frac{1}{n_g} \sum_{k=1}^{n_g} \tilde{S}(\mathbf{x}_c^n + \rho^n \mathbf{n}_k - \mathbf{x}_{i+1/2}).
\end{equation}

However, the logic above implicitly assumes common definitions of $n_g$ and $\mathbf{n}_k$ at time-step $n$ and $n+1$.  Our proposal to adaptively choose $n_g$ based on local field conditions is inconsistent with this assumption.  The topic of exact charge conservation with adaptive gyro-averaging thus represents a challenge for future work in applying this scheme to self-consistent PIC simulations.

\subsection{Alternating Time-steps} \label{sec:altdt}
In developing the FLR-capturing scheme above, the authors attempted several other solutions that, while ultimately unsatisfactory for realistic field configurations, led to a useful technique for relaxing time-step restrictions on the scheme that we describe here.  We begin in a simplified scenario with fixed, spatially uniform $\bE$ and $\bB$.  We will argue later that our conclusions extend to more general field configurations.  We recall from the analysis in \cite{ricketson2019energy} -- which borrows from ideas in \cite{parker1991numerical} -- that the perpedicular velocity update for Crank-Nicolson (and thus for \eqref{eq:CBFVdef} in uniform $\bB$) may be written
\begin{equation}
	\bv_\perp^{n+1} = \bv_E + \mathcal{R}_\theta \left[ \bv_\perp^n - \bv_E \right],
\end{equation}
where $\bv_E$ and $\mathcal{R}_\theta$ are as defined above, and the rotation angle $\theta$ is again defined by \eqref{eq:rotangle}.

In this simple circumstance, the gyration velocity $\bu$ is exactly $\bv_\perp - \bv_E$, so this expression may be shortened to
\begin{equation}
	\bu^{n+1} = \mathcal{R}_\theta \bu^n.
\end{equation}
We now propose to alternate between a large time-step $\Dt$ and a smaller one $\dt$.  Let these time-steps give rise to gyration velocity rotation angles $\theta$ and $\delta \theta$, respectively.  Without loss of generality, let $t^{n+1} - t_n = \Dt$ and $t^{n+2} - t^{n+1} = \dt$.  It is straightforward to see, geometrically, that the angle between $\bu\nph$ and $\bu^{n+3/2}$ is $(\theta + \delta \theta)/2$.  This observation is independent of the ordering of the large and small steps, so we see that the half-step gyration velocity rotates by $(\theta + \delta \theta)/2$ at \textit{every} time-step.  

It is this simple observation that motivates our alternating-step approach.  Note that as a consequence of \eqref{eq:rotangle}, for $\Oc \Dt \gg 1$ we have $\theta \approx \pi$.  When stepping over the gyration time-scale by a large factor, a particle thus oscillates between two nearly diametrically opposed gyrophases at each time-step, meaning it only very slowly accumulates a representative sample of phases from the entire orbit.  Noting, however, that the half-step velocity is that which actually updates the particle's position [see \eqref{eq:CBFVdef}] we see that alternating time-steps allows us some control over the displacement rotation angle, and thus over which gyrophases are sampled.  While the scheme now explicitly computes a gyro-average of $\mathbf{E}$, it still relies on a sequence of consecutive time-steps to approximate a gyro-average of $\mathbf{F}_{eff}$.  If we can accumulate a representative gyrophase sampling more quickly by controlling gyrophase traversal, one expects this to allow us to take larger time-steps.  

We describe how to exercise this control to achieve arbitrary numbers of equispaced gyrosamples in what follows.  First, though, it is important to confirm that this new alternating-step strategy does not break the desirable properties enjoyed by \eqref{eq:CBFVdef}.  To that end, recall that the geometric derivations in \cite{birdsall2004plasma, parker1991numerical, ricketson2019energy} of the numerical gyroradius of the Boris and Crank-Nicolson integrators assume a fixed time-step.  Indeed, they assume the particle displacement vector has fixed magnitude at each step.  This assumption is clearly violated by the proposed alternating time-step scheme, so it is not obvious that Crank-Nicolson (and thus the AP integrator \eqref{eq:CBFVdef} above) still recovers the correct gyroradius when alternating between large and small time-steps.  

Admittedly, it would be surprising if the numerical gyroradius deviated significantly from the physically correct gyroradius as a result of varying time-step size.  Nevertheless, it is a valuable exercise to show that, in fact, the \textit{exact} correct gyroradius is still recovered for an \textit{arbitrary} sequence of time-steps.  This is shown in Lemma 1 below, and the tools used in this proof will be instrumental in deriving new time-step restrictions on the scheme.  

\begin{lemma}
	Let the particle position and velocity be updated according to 
	\begin{equation} \label{eq:gmotiononly}
	\begin{split}
		\bx^{n+1} &= \bx^n + \Dt_n \bu\nph, \\
		\bu^{n+1} &= \bu^n + \frac{q}{m} \Dt_n \bu\nph \times \bB,
	\end{split}
	\end{equation}
	with $\bB$ constant in space and time, and for an arbitrary sequence of time-steps $\Dt_n$. Further, let $\bu^0 \cdot \bB = 0$.  Then \textit{all} particle locations $\bx^n$ lie on a single circle of radius $\rho = u/\Oc$, where $u$ is the common magnitude of all the vectors $\bu^n$.  
\end{lemma}

\begin{proof}
	See Appendix A.
\end{proof}

It is also important to understand the relationship of the numerical gyrocenter to the analytic gyrocenter.  Indeed, simply having the correct gyro\textrm{radius} is insufficient to guarantee the correct gyro-\textrm{orbit}.  To that end, we note a surprising consequence of Lemma 1:  Under the evolution \eqref{eq:gmotiononly}, the first three particle locations $\{ \bx^0, \bx^1, \bx^2 \}$ define a unique circle.  According to Lemma 1, this circle has radius $\rho$ and \textit{all} other future particle locations lie on it.  Thus, in this simple scenario, the gyro-orbit -- including the gyro-center -- is completely specified by the first two time-steps, independent of the size of all future time-steps!

This simple observation is the basis for the following corollary.
\begin{cor}
Let particle position and velocity again be updated according to \eqref{eq:gmotiononly} with $\bu^0 \cdot \bB = 0$.  As before, assume that the sequence $\Dt_n$ is arbitrary and $\bB$ is independent of $\bx$ and $t$.  Then the particle's gyrocenter -- and thus the entire gyro-orbit -- is recovered exactly in the two limits 
	\begin{enumerate}[label=(\roman*)]
		\item $\displaystyle \Oc \max\{ \Dt_0, \Dt_1 \} \rightarrow 0$,
		\item $\displaystyle \Oc \Dt_0 \rightarrow \infty$.
	\end{enumerate}
\end{cor}
\begin{proof}
See Appendix B.
\end{proof}
We thus see that if one wishes to recover the gyro-orbit accurately, one only needs to place a constraint on the initial time-step(s).  One may either initialize a simulation with two very small time-steps satisfying $\Oc \max \{ \Dt_0, \Dt_1 \} \ll 1$, or one may insist on a very large initial step satisfying $\Oc \Dt_0 \gg 1$.  

Having confirmed that we do not sacrifice our ability to capture gyro-motion by varying $\Dt$ from time-step to time-step, we return to our strategy of alternating large steps $\Dt$ with smaller steps $\dt$.  As already noted, this results in the half-step gyration velocity rotating by angle $\psi = (\theta + \delta \theta)/2$ at every time-step.  Under pure gyration, the angle connecting $\bx^n \rightarrow \bx^{n+1} \rightarrow \bx^{n+2}$ is thus $\pi - \psi$ and is inscribed in the circle of radius $\rho$ on which these three particle positions lie.  By the inscribed-angle theorem, the arc connecting $\bx^n \rightarrow \bx^{n+2}$ subtends a central angle $2(\pi - \psi)$.  This state of affairs is illustrated in Figure \ref{fig:inscangle}.

\begin{figure}[h]
	\centering
	\includegraphics[width=0.5\textwidth]{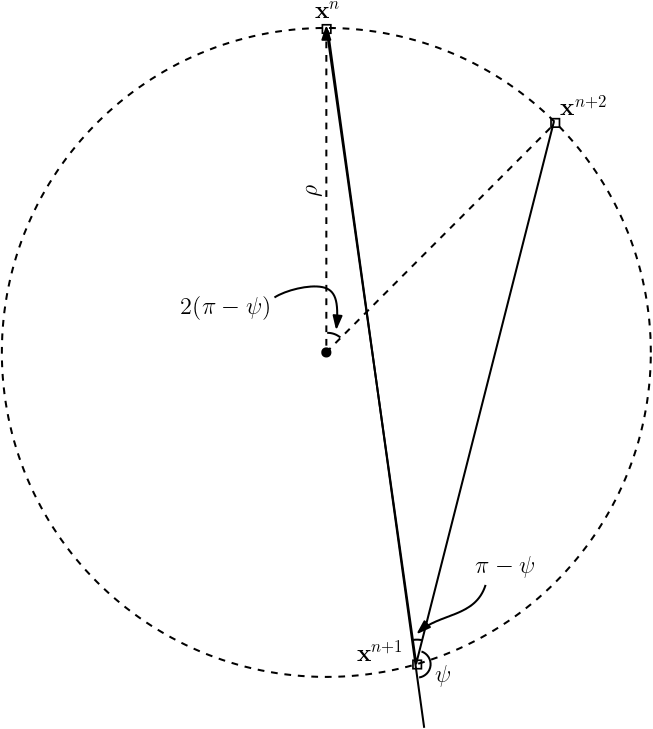}
	\caption{Illustration of the relationship between half-step velocity rotation angle between two time-steps - i.e. $\psi$ - and the gyrophase angle subtended by the resulting particle positions.} \label{fig:inscangle}
\end{figure}

We find it intuitive to think in terms of asking the particle to sample a number $m$ of equispaced gyrophases over a sequence of time-steps.  This $m$ need not be the same as $n_g$.  Indeed, it is not even necessary that $m$ be an integer.  Ultimately, we will choose $m$ to make the time-step constraints that result from this alternating-step strategy as loose as possible.   Using the above results, we see that if we wish for the 
%even (respectively, odd) 
time-step particle locations to sample $m$ evenly spaced points around the orbit, we must require $2(\pi - \psi)m = 2\pi$.  In this way, we traverse the entire circle over $m$ 
%even (odd) 
time-steps.  Solving this expression for $\delta \theta$ gives:
\begin{equation} \label{eq:dthetaform}
	\delta \theta = \left( 1 - \frac{1}{m} \right) 2 \pi - \theta.
\end{equation}
Recall that the large time-step $\Dt$ sets the large rotation angle $\theta$.  Given that, the formula above tells how to pick the small time-step.  Indeed, \eqref{eq:rotangle} may be inverted to find:
\begin{equation} \label{eq:rotangleinverse}
	\Oc \dt = 2 \sqrt{ \frac{ 1 - \cos \delta \theta }{ 1 + \cos \delta \theta } },
\end{equation}
so the angle $\delta \theta$ uniquely specifies the small time-step value.  It is instructive to substitute \eqref{eq:dthetaform} into \eqref{eq:rotangleinverse}, yielding:
\begin{equation} \label{eq:Ocdt_def}
	\Oc \dt = 2 \sqrt{ \frac{ 1 - \cos \left( \theta + 2\pi/m \right) }{ 1 + \cos \left( \theta + 2\pi/m \right) } }.
\end{equation}
One can in theory use trigonometric identities and \eqref{eq:rotangle} to directly write $\Oc \dt$ in terms of $\Oc \Dt$, but the expression is cumbersome and unenlightening.  However, we note that in the limit $\Oc \Dt \gg 1$, we have $\theta \approx \pi$, and therefore: 
\begin{equation}
	\Oc \dt \approx 2 \sqrt{ \frac{ 1 + \cos \left( 2\pi/m \right) }{ 1 - \cos \left( 2\pi/m \right) } }.
\end{equation}
Taylor-expanding the cosine, we find:
\begin{equation} \label{eq:approxdt}
	\Oc \dt \approx \frac{2}{\pi} m,
\end{equation}
which is reasonably accurate for $m \geq 8$.  While these limits are useful in building intuition, we use the exact expression (Eq. \ref{eq:Ocdt_def}) when implementing the scheme.  

%%%%%%%%%%%%%%%%%%%%%%%%%%%%%%%%%%%%%%%%%%%%%%%%%%%

\subsection{Time-step Constraints}
As ever, the ability to step over the gyration time-scale does not permit arbitrarily large time-steps.  An understanding of the restrictions on time-step size is crucial in practice to avoid accepting erroneous results as legitimate.  In this section, we extend the results on time-step restrictions developed for the previous scheme \eqref{eq:CBFVdef} in \cite{ricketson2019energy}.  These prior results were summarized in Table 1.  

Here, we will restrict ourselves to the case $u\nph \gg v_E$ when $k_\perp \rho \sim 1$.  This is so because we require $v_E \Dt \ll k_\perp^{-1} \approx \rho$ to resolve the $\bE \times \bB$ drift motion, and $\Omega_c \Dt \gg 1$ implies $u\nph \Dt \approx 2\rho$, and therefore $u\nph \gg v_E$. Indeed, the very assumption that the gyro-orbit is circular (rather than elliptical) implies that the gyration velocity is large compared to the $\bE \times \bB$ drift velocity.  Promising steps to relax this assumption have been taken in \cite{joseph2021guiding}, but we only consider here the standard gyrokinetic approximation.  In cases with $k_\perp \rho \ll 1$, the time-step constraints previously derived in \cite{ricketson2019energy} in the limit $u^{n+1/2} \ll v_E$ still apply.

As in \cite{ricketson2019energy}, there are two constraints one must obey if there is to be hope of accurate integration.  First, the time required to accumulate an accurate gyro-average of $\mathbf{F}_{eff}$ must be small compared to time-scales of interest, denoted $\tau_{res}$.  Second, the anomalous displacement due to imperfect gyro-averaging on shorter time-scales should be small compared to length scales of interest, denoted here by $L_{res}$.  

The mathematical formulation of the first constraint is considerably simplified for the new scheme due to the alternating large and small time-steps.  The time required to accumulate an accurate gyro-average is simply $m (\Dt + \dt)/2$.  We thus simply have the constraint:
\begin{equation} \label{eq:timerestric}
	m \frac{\Dt + \dt}{2} < \tau_{res}.
\end{equation}
Of course, if $\Oc \tau_{res} \gg 1$, this still permits $\Oc \Dt \gg 1$ for modest $m$ values.  Indeed, isolating $\Oc \Dt$ and using the approximate value of $\dt$ from \eqref{eq:approxdt}, this constraint may be simplified to
\begin{equation} \label{eq:timerestricsimp}
	\Oc \Dt < \frac{2}{m} \Oc \tau_{res} - \frac{2}{\pi}m.
\end{equation}

More complex is the derivation of time-step restrictions due to anomalous displacements arising from short-time-scale errors in gyro-averaging.  As before, such anomalous displacements arise from the fact that the $\nabla B$ force only approximates the physically correct value in a gyro-averaged sense, and thus deviates from it on short time-scales, leading to anomalous displacement.  
Again following \cite{ricketson2019energy}, we recall that the $\nabla B$ drift has the physical velocity
\begin{equation}
	\bv_{\nabla B} = \frac{\bb}{\Oc} \times \frac{\mu}{m} \nabla B, 
\end{equation}
while the AP scheme \eqref{eq:CBFVdef} features the drift velocity
\begin{equation}
	\bv^\text{AP}_{\nabla B} = \frac{\bb}{\Oc} \times \frac{1}{m} \left( \mu_\text{eff}^\text{CN} \nabla B - \bF_{eff} \right),
\end{equation}
where $\mu_\text{eff}^\text{CN} = \mu / (1 + \Oc^2 \Dt^2 / 4)$.  The scheme thus features an anomalous displacement with velocity
\begin{equation}
	\Delta \bv_{\nabla B} = \bv_{\nabla B} - \bv^\text{AP}_{\nabla B}.
\end{equation}
By virtue of the scheme's construction, $\left\langle \Delta \bv_{\nabla B} \right\rangle = 0$ to a good approximation, but it may deviate from zero on shorter time-scales.  In Appendix C, we show that this displacement may be analyzed in the presence of alternating large and small time-steps using the tools developed in Appendix A.  The end result of the analysis in Appendix C is another constraint on time-step: if the anomalous displacement's magnitude is to be bounded above by $L_{res}$, then one requires
\begin{equation} \label{eq:spacerestric}
	\Oc \left( \frac{\Dt + \dt}{2} \right) < \min \left\{ \frac{L_{res}}{\rho} \frac{\sin \left( 2 \pi / m \right)}{\delta_\perp}, \left[ \frac{L_{res}}{\rho} \frac{\sin \left( \pi / m \right)}{\delta_\parallel} \right]^{1/2} \right\}.
\end{equation}

Typically, it will be the case that $L_{res} \approx k_\perp^{-1}$.  Additionally, $\delta_\perp$ and $\delta_\parallel$ are as defined in \eqref{eq:deltadefs}.  Thus, we note that we are still permitted to step over the gyration time-scale even when $k_\perp \rho \sim 1$ so long as $\delta_\perp$ and $\delta_\parallel$ are small.  Again, if we assume $\Oc \Dt \gg 1$, we can use our approximation of $\Oc \dt$ to write a restriction on the large time-step $\Dt$ alone:
\begin{equation} \label{eq:spacerestricsimp}
	\Oc \Dt < 2\min \left\{ \frac{L_{res}}{\rho} \frac{\sin \left( 2 \pi / m \right)}{\delta_\perp}, \left[ \frac{L_{res}}{\rho}\frac{\sin \left( \pi / m \right)}{\delta_\parallel} \right]^{1/2} \right\} - \frac{2}{\pi} m.
\end{equation}

Note that the restrictions \eqref{eq:timerestric} and \eqref{eq:spacerestric} each represent an improvement over the previously derived restrictions in Table 1.  Previously, time-steps were bounded above by multiples of $\sqrt{\Oc \tau_{res}}$ and $1/\sqrt{\delta_\perp}$, while now they are bounded by multiples of $\Oc \tau_{res}$ and $1/\delta_\perp$.  This arises directly from the improved gyro-averaging, which for very large time-steps shortens the number of steps required to accumulate a representative sample of gyrophases.  In typical, strongly magnetized applications, we expect $ \Oc \tau_{res}, \delta_\perp^{-1} \gg 1$, so this improvement can lead to significant computational gains.  

Of course, we still find that the time-step is constrained to be $O(1/\sqrt{\delta_\parallel})$.  Thus, in situations in which parallel motion represents the active constraint on time-step size, no improvement has been achieved.  However, it is typically the case in fusion plasmas that characteristic length scales in the parallel direction are much longer than those in the perpendicular direction.  One thus expects it to be somewhat rare that resolution of parallel motion represents a more stringent constraint than perpendicular motion, and there is thus reason to hope that the relaxed constraints in the perpendicular direction will be applicable to the majority of particles in typical simulations.  

%%%%%%%%%%%%%%%%%%%%%%%%%%%%%%%%%%%%%%%%%%%%%%%%%%%
\subsection{Adaptive Time-Stepping and Gyro-averaging}
\label{sec:alt-dt}
It is clearly desirable to select the largest possible time-step allowed by the constraints above, so that the orbit may be captured accurately with minimal computational cost.  In this section, we present a slight modification of the adaptive time-stepping scheme of \cite{ricketson2019energy}.  The modifications are motivated by two considerations: (a) the relaxed time-step constraints of the new scheme derived in the section above, (b) the simplifying assumption that $v_E^{n+1/2} \ll u^{n+1/2}$ is, as argued above, necessary for stepping over the gyroperiod when $k_\perp \rho \sim 1$.  Future work will seek to derive a more general adaptive time-stepping strategy that smoothly transitions between gyrokinetic, drift-kinetic (in which it may be true that $u^{n+1/2} \leq v_E$), and unmagnetized regimes.  Additionally, the adaptive gyro-averaging in the new scheme requires an estimate of a local value of $k_\perp$ in order to choose the number of gyrophase samples $n_g$.  There are many possible methods to achieve this, but we choose a particularly simple one below.  

In principle, adaptive time-step selection is a simple process.  One estimates local length scales, both in the parallel and perpendicular directions, as well as particle velocity in both directions.  Dividing a characteristic length by a velocity gives a characteristic time.  One then chooses a time-step that is small compared to both parallel and perpendicular characteristic times, and also respects the time-step constraints derived in the previous section.  The challenge, of course, is in the details.  

First, we require an estimate of $k_\perp$ to use \eqref{eq:adaptive_ng} to specify a a number of gyro-average samples.  The simplest method would be to estimate $k_\perp$ via $\| \nabla_\perp \mathbf{E} \| / E$, but this of course fails at local extrema of the electric field.  We thus use the somewhat more robust estimate
\begin{equation}
	k_\perp = \max \left\{ \frac{\| \nabla_\perp \mathbf{E} \|}{E}, \left( \frac{ \| D^2 \mathbf{E} \| }{E} \right)^{1/2} \right\},
\end{equation}
where $\| D^2 \mathbf{E} \|$ is defined by
\begin{equation}
	\| D^2 \mathbf{E} \|^2 = \| \partial_x^2 \mathbf{E}\|^2 + \| \partial_y \mathbf{E} \|^2,
\end{equation}
where $x$ and $y$ here are Cartesion coordinates orthogonal to $\bB$ at the particle location.  In our code, $x$ is oriented in the $\bE \times \bB$ direction, and $y$ is in the $(\bE \times \bB) \times \bB$ direction.  Intuitively, this procedure uses both first and second derivatives to estimate $k_\perp$ and trusts whichever estimate is more pessimistic.  Of course, there are numerous methods to estimate characteristic perpendicular length scales in $\bE$, each with benefits and drawbacks.  Implementation of this method in self-consistent PIC codes in the future will shed additional light on appropriate choices.  

We now consider the three free parameters in Eq. \eqref{eq:adaptive_ng}.  Note that when the nondimensional step size $\Oc \Dt \leq b^{-1}$, we are guaranteed $n_g = 1$.  One expects to need no gyro-averaging when the gyroperiod is reasonably well resolved, so we choose $b^{-1} = 1/2$.  For large time-steps, $a$ sets the number of gyro-samples we use when $k_\perp \rho = 1$.  We find $a=16$ samples to be reasonably accurate for these parameters. The parameter $c$ sets the sensitivity of $n_g$ to variations in $k_\perp$.  Because our estimate of $k_\perp$ may vary somewhat from time-step to time-step, we find it desirable to keep $n_g$ more consistent in size by choosing $c=1/2$.  Finally, we impose the absolute maximum value of $n_g$ to be 64.  

We next discuss how we estimate characteristic length scales.  We choose
\begin{equation}
\begin{split}
	L_{res,\parallel} &= \Gamma \min \left\{ \frac{B}{| \nabla_\parallel B | }, \kappa^{-1}, \frac{v_E}{| \nabla_\parallel v_E |} \right\}, \\
	L_{res,\perp} &= \Gamma \min \left\{ \frac{B}{\| \nabla_\perp B \|}, \frac{v_E}{\| \nabla_\perp v_E \|} \right\}.
\end{split}
\end{equation}
Here, $\kappa = \| (\bb \cdot \nabla) \bb \|$ is the local magnetic field curvature.  $v_E$ is computed using a gyro-averaged electric field, so this estimate does not explicitly include any gyro-scale structure of $\mathbf{E}$.  It is expected that such structure is captured by the adaptive gyro-averaging procedure.  $\Gamma$ is a free parameter that governs what fractional change in each quantity we are willing to tolerate within a time-step.  We follow \cite{ricketson2019energy} by choosing $\Gamma = 0.1$.  

Pseudo-code for the adaptive time-stepping procedure is as follows:
\begin{enumerate}
	\item Compute $k_\perp$, $n_g$, $L_{res,\parallel}$, and $L_{res,\perp}$ as described above.  
	\item Let 
	\begin{equation}
		\tau_{res} = \min \left\{ \frac{L_{res,\perp}}{v_E}, \frac{L_{res,\parallel}}{v_\parallel} \right\}.
	\end{equation}
	\item Set $F = (L_{res,\perp}/\rho) \sin \left( 2 \pi/m \right)$, and use time-step
	\begin{equation}
		\Dt = 2 \Omega_c^{-1} \alpha \min \left\{ \frac{F}{\delta_\perp}, \sqrt{\frac{F}{\delta_\parallel}}, \frac{\Omega_c \tau_{res}}{m} \right\}
	\end{equation}
	%\item Use \eqref{eq:Ocdt_def} and \eqref{eq:rotangle} to compute the subsequent smaller time-step $\delta t_S$ corresponding to the large time-step $\Delta t_S$.  
	%\item Use time-step $\max \{ \Delta t_S, \delta t \}$.  
\end{enumerate}
As in \cite{ricketson2019energy}, $\alpha$ is a parameter that controls how close to the time-step restrictions we are willing to approach.  As before, we use $\alpha = 0.9$.  

Clearly, this procedure is only used to pick the \textit{larger} time-steps $\Delta t$ in the alternating time-step procedure.  Note that if $\Delta t$ is sufficiently small, it is actually possible that the ``smaller" time-step $\delta t$ defined by \eqref{eq:Ocdt_def} and \eqref{eq:rotangle} can be larger than $\Delta t$, and thus inadmissible.  For this reason, at the subsequent time-step, we will choose $\min \{ \Delta t, \delta t \}$, where $\Delta t$ is specified using the procedure above.  

In this initial proof-of-principle study, we additionally specify an absolute cap on the size of $\Omega_c \Dt$.  In our tokamak test case below, this cap is set to 70, but results are not substantially impacted by small changes to this parameter.  

%%%%%%%%%%%%%%%%%%%%%%%%%%%%%%%%%%%%%%%%%%%%%%%%%%%

\subsection{Algorithm Summary and Complex Field Configurations} \label{sec:summary}

In this section, we outline the proposed algorithm and provide an intuitive justification for its use in complex field configurations.  The overall structure of the procedure is as follows:

\begin{enumerate}
	\item Given $\bx^n$, $\bv^n$, choose an appropriate time-step $\Dt$ and number of gyrophase samples $n_g$ using the procedure in the previous subsection.  
	\item Advance in time by one step of size $\Dt$ using the scheme \eqref{eq:CBFVdef}, but with the modification that 
	\begin{equation} \label{eq:expgavgE}
		\bE\nph = \frac{1}{n_g} \sum_{k=1}^{n_g} \bE \left( \bx_c^{n+1/2} + \rho^{n+1/2} \mathbf{n}_k, t^{n+1/2} \right).
	\end{equation}
	\item Now, having $\bx^{n+1}$, $\bv^{n+1}$, choose the next time-step $\dt$ satisfying
	\begin{equation}
		\Oc^{n+1} \dt = 2 \sqrt{ \frac{ 1 - \cos \left( \theta + 2\pi/m \right) }{ 1 + \cos \left( \theta + 2\pi/m \right) } }
	\end{equation}
	where $\theta$ is defined by $\cos \theta = (1 - (\Oc^n)^2 \Dt^2/4)/(1 + (\Oc^n)^2 \Dt^2/4)$, and $\Oc^n = q B(\bx^n)/m$.  
	\item Advance by one step of size $\dt$, again using \eqref{eq:CBFVdef} with the modification \eqref{eq:expgavgE}.  
	\item $n \leftarrow n + 2$ and return to step 1.
\end{enumerate}

%Under the assumption that $\bB$ varies on spatial scales much longer than the gyro-radius -- i.e. $\delta_\perp \ll 1$ -- and on temporal scales much longer than $\Dt$, $\Oc^n$ and $\Oc^{n+1}$ are comparable and this procedure results in approximately equispaced gyrosamples, even in varying $\bB$.  

%The gyrokinetic assumption that one can gyro-average over a circular orbit requires that the gyro-motion be faster than any other particle motion, including the $\bE \times \bB$ drift.  This is in contrast to drift-kinetic, guiding center motion derivations -- e.g. \cite{hazeltine2003plasma, hazeltine2018framework} -- that allow gyration and $\bE \times \bB$ velocities to be comparable.  Thus, the constraint imposed that the time-step resolves the spatial variation in the $\bE$-field, $v_E \Dt \ll k_\perp^{-1}$, is really a constraint of the gyrokinetic approximation itself.  

We note that in our derivations we took advantage of the fact that the gyration velocity $\bu$ is exactly $\bv - \bv_E$.  This is still true to good approximation in varying fields, for all other drift motions are small in the magnetized regime.  When these drift motions grow to a substantial size (as can easily be tested at each step based on local field quantities), the gyro-motion must be resolved in any case.  

%Finally, we note that the role of $m$ is different depending on whether one is using Method 1 of Section 3.2.1 or Method 2 of Section 3.2.2.  In the former case, $m$ represents the number of gyrophase samples used to approximated $\langle \bE \rangle$ over a sequence of consecutive time-steps, and should thus be chosen large enough to resolve spatial structures in $\bE$.  In contrast, for Method 2, the approximation of the gyroaveraging operator is managed at each time-step with $n_g$ gyrosamples, and it is not necessary that $n_g$ be equal to $m$.  In the context of Method 2, then, $m$ merely plays the role of setting the size of the small time-steps $\delta t$.  

One should choose $m$ so that time-steps can be taken as large as possible on average.  Even casual investigation of \eqref{eq:timerestric} and \eqref{eq:spacerestric} reveals that the optimal $m$ that allows the largest possible average time-step depends on which constraint is active.  We take a simplified approach and uniformly choose $m=5$.  

%%%%%%%%%%%%%%%%%%%%%%%%%%%%%%%%%%%%%%%%%%%%%%%%%%%

\section{Numerical Results}
The new scheme is tested in four configurations.  The first features constant $\bB$, with $\bE$ varying sinusoidally in one coordinate.  This case is chosen for the ability to compare with asymptotic theory predicting the $\bE \times \bB$ drift velocity as a function of $k_\perp \rho$ for $k_\perp \rho \lesssim 1$.  In the second case, we let $\bE$ vary sinusoidally in both coordinates and confirm that the alternating time-step methodology permits enlarged time-steps. The third case features linearly varying $\bB$ with $\bE$ again varying sinusoidally in both perpendicular coordinates, although this time with $k_\perp \rho > 1$.  The fourth case uses a tokamak geometry with Solov'ev equilibrium $\bB$ field with a rapidly varying perpendicular $\bE$, chosen to test the scheme's ability to capture parallel dynamics in conjunction with complex geometry and FLR effects.  

For clarity, we work in a dimensionless formulation in which $q_p/m_p = 1$ (implying $B = \Oc$).  To measure the accuracy of the present scheme (which we term CBFV+FLR), we compare with asymptotic theory when available and with the results of the widely-used Boris integrator using $\Oc \Dt = 0.1$.  We also compare against other AP integrators such as the scheme of Brackbill, Forsluand and Vu (BFV) \cite{brackbill1985simulation, vu1995accurate}, the present scheme's precursor (CBFV) \cite{ricketson2019energy}, and the standard Crank-Nicolson (CN) scheme.

The systems of equations resulting from our implicit integrator are solved using the preconditioned Picard iteration introduced in \cite{koshkarov2022fast} with nonlinear tolerance set to $10^{-12}$.  In the first three cases, the particle is initialized in the $x$-$y$ plane at
\begin{equation} \label{eq:initialization}
	\bx(t=0) = \frac{1}{B} \left( \begin{array}{c} \cos \omega \\ -\sin \omega \end{array} \right), \qquad \bv(t=0) = \left( \begin{array}{c} -\sin \omega \\ -\cos \omega \end{array} \right)
\end{equation}
for some $\omega$, and the $B$ above is evaluated at the origin.  This is done so that the particle's initial gyrocenter is located at the origin.  While we do not report them here in the interest of brevity, all tests are performed for multiple values of $\omega$ to confirm the gyrophase-independence of the results.  

The code is implemented in Python, and beyond the use of the efficient solver of \cite{koshkarov2022fast}, no effort has been made at performance optimization in this proof-of-principle study.  This has two noteworthy consequences.  Firstly, direct computation-time comparisons between the new method and Boris (or any other method) are not to be trusted if one wishes to extrapolate to optimized, compiled code running on HPC systems.  We thus do not report specific numbers, instead only noting that our implicit implementation is typically between 10$\times$ and 120$\times$ faster than well-resolved Boris in the experiments performed, lending confidence that considerable speed-ups are possible.  Note that in practice, speed-ups will be problem-dependent, and in particular depend on the non-dimensional parameters appearing in \eqref{eq:deltadefs}.

Secondly, the exact energy conservation properties the scheme provably enjoys in the implicit-PIC context do not carry over to single particle motion in prescribed analytic fields.  Indeed, the Crank-Nicolson scheme on which it is based only conserves energy to $O(\Dt^2)$ in analytically prescribed fields, and our modifications do not change that fact.  In \cite{ricketson2019energy}, this was corrected using the scheme of \cite{simo1992exact}, in which the electric field is modified according to
\begin{equation}
	\bE\nph \rightarrow \frac{\phi^n - \phi^{n+1}}{\left( \bx^{n+1} - \bx^n \right) \cdot \bE\nph} \bE\nph
\end{equation}
to recover exact energy conservation, where $\phi$ is the electrostatic potential associated with $\bE$.  However, in the highly oscillatory electric fields used here, the exactly-conservative scheme has tremendous difficulty converging without sophisticated preconditioning.  As such, when measuring energy conservation, we will be satisfied with observing that energy errors are small and do not feature secular growth in time.  Future work will empirically confirm whether exact conservation is realized when coupling to an implicit PIC solver.

\subsection{1-D Sinusoidal Variation featuring asymptotic theory}
We set $\bB = 100 \widehat{\mathbf{z}}$, with 
\begin{equation}
	\bE = \cos (k_\perp y) \widehat{\mathbf{y}}.
\end{equation}
The particle is initialized with $\omega = \pi/2$.  Asymptotic theory \cite{currelinotes} predicts that the $\bE \times \bB$ drift velocity in this scenario is given by
\begin{equation}
	\bv_E \approx \bv_{E,gc} \left( 1 - \frac{k_\perp^2 \rho^2}{4} \right),
\end{equation}
where $\bv_{E,gc}$ is the drift velocity computed with $\bE$ at the gyrocenter, i.e., the drift-kinetic limit.  This expression, being derived via Taylor expansion in $k_\perp \rho$, is valid for $k_\perp \rho \lesssim 1$.  

As the particle is initialized with its gyrocenter at the origin (where $\bE = \widehat{\mathbf{y}}$) we have $\bv_{E,gc} = \widehat{\mathbf{x}}/100$.  We simulate to a final time $T=100$, so that the predicted final $x$-position of the particle is simply $x(T) = 1 - k_\perp^2 \rho^2/4$.  This asymptotic theory result is compared in Figure \ref{fig:asymptheory} to Boris, CBFV, and the new scheme (CBFV+FLR).  For this simple test, we do not use adaptive time-stepping or adaptive gyro-averaging.  The larger time-step is fixed at $\Omega_c \Dt = 100$, with alternating steps chosen in the manner described in Section 3.4.  We show results using fixed values of $n_g$ of 8 and 4.   
\begin{figure}
	\centering
	\includegraphics[width=0.6\textwidth]{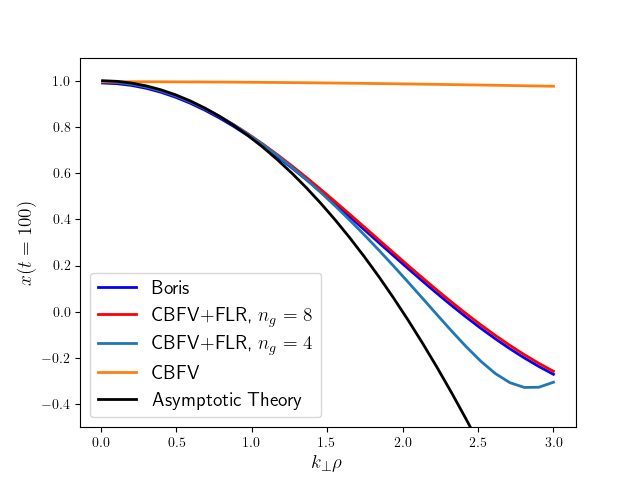}
	\caption{Final $x$-position of the particle resulting from Boris (blue) integration, the new scheme with $n_g = 8$ gyrophase samples (CBFV+FLR, red), new scheme with $n_g = 4$ (CBFV+FLR, teal), asymptotic theory (black), and a prior AP scheme with no FLR corrections (CBFV, orange).  We see the ability of the new scheme to match asymptotic theory when $k_\perp \rho \lesssim 1$, as well as good agreement with fully-resolved Boris out to even large $k_\perp \rho$ when 8-point gyroaveraging is used.}
	\label{fig:asymptheory}
\end{figure}

One immediately sees agreement with the asymptotic theory in its range of validity, $k_\perp \rho < 1$.  Outside that range, the well-resolved Boris simulation may be regarded as the ground truth.  We thus also see excellent agreement of new scheme well beyond $k_\perp \rho = 1$ when $n_g = 8$.  As is to be expected, more accurate gyro-averaging is necessary to improve agreement for large $k_\perp \rho$.  Also as expected, the prior scheme without FLR corrections fails to predict any change in the drift velocity as a function of $k_\perp$.  

\subsection{Alternate Time-stepping test}
Next, we perform a similar test to the last one, but focusing on the benefits of the proposed strategy of alternating large and small time-steps.  We use fields with locally similar qualitative properties but different parameters:
\begin{equation}
	\mathbf{B} = \left( 100 + 3 \sin x \right) \widehat{\mathbf{z}}, \qquad \mathbf{E} = \frac{1}{10} \left( \cos k_x x + \sin k_y y \right).
\end{equation}
We choose $k_x = 5\pi$ and $k_y = 1$.  We initialize the particle at the origin with $\mathbf{v}^0 = \widehat{\mathbf{x}}$.  
These parameters result in $k_\perp \rho \approx 0.15$, which is an optimal regime to demonstrate the benefits of the alternate-stepping strategy.
%We have chosen these parameters specifically to emphasize the benefits of the alternate-stepping approach.  Indeed, it is shown in \cite{ricketson2019energy} that when the scheme does \textit{not} alternate steps, the particle's orbit traverses another circle -- superimposed on the gyromotion -- whose radius may exceed the gyroradius if the time-steps are taken too large.  While this is clearly physically incorrect, it is not catastrophic on the particle orbit as long as $k_\perp \rho \ll 1$.  Thus, to observe tangible benefits of alternating time-steps, we concentrate on non-negligible $k_\perp \rho$.  However, we wish to emphasize that the benefits of alternate-stepping, while synergistic with the gyroaveraging of the electric field, and independent of it.  \lc{I can't parse the last sentence.}To do this, we desire an example in which $k_\perp \rho$ is not \textit{too} large.  The value of $0.15$ chosen for this test does a reasonable job of achieving both goals.  

Because we are specifically interested in pushing the limits of the time-step size in the alternating-step method, for this test we do not use the adaptive time-stepping scheme described in Section \ref{sec:alt-dt}.  Instead, we manually choose the large time-step as $\Omega_c \Delta t = 800$, while the small step $\delta t$ is chosen as described in Section \ref{sec:alt-dt}.  When not using the alternating-step technique, we fix the time-step at $\Omega_c \Delta t = 400$ so that the the two methods have a roughly equal average time-step.  

We plot the resulting orbits in the $x$-$y$ plane as well as their time histories in Figure \ref{fig:AltTest}.  It is immediately evident that the alternating-step approach dramatically improves the accuracy of the orbit, as predicted.  Indeed, the correct orbit is periodic, while without alternating steps the orbit is non-periodic on the time-scales observed.  

\begin{figure}[h!]
	\centering
	\includegraphics[width=0.6\textwidth]{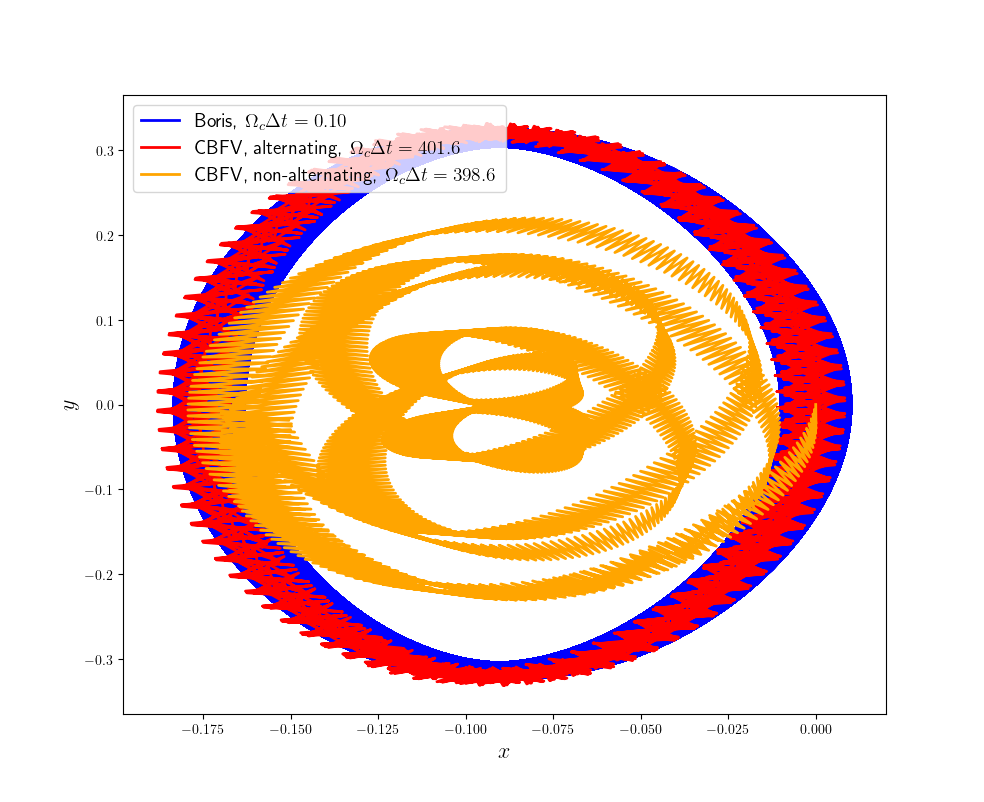}
 \includegraphics[width=0.8\textwidth]{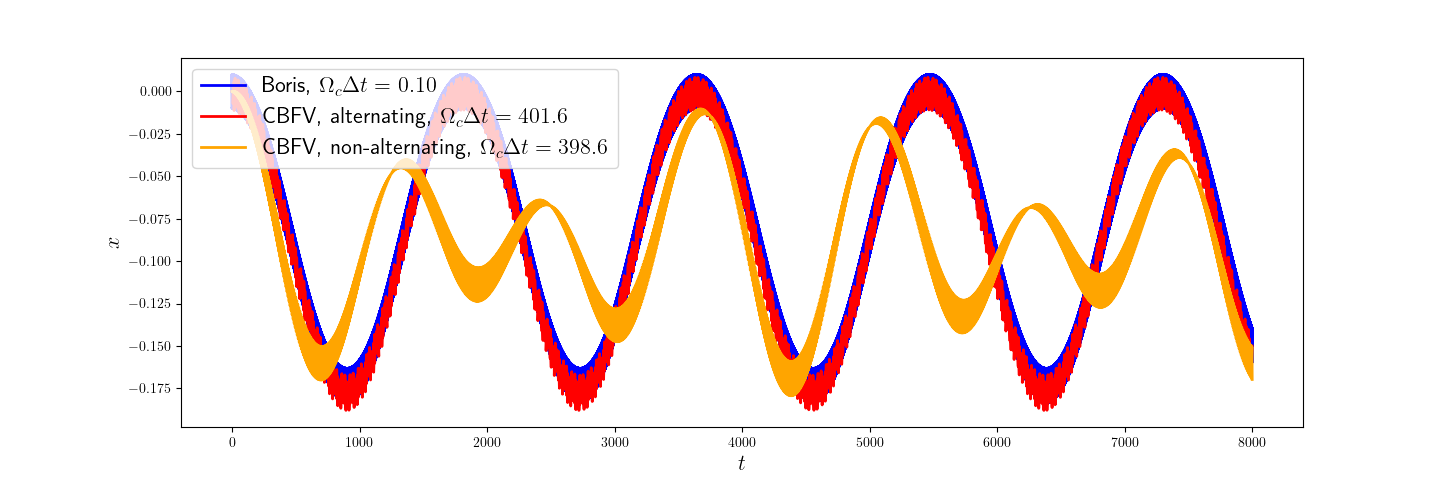}
	\includegraphics[width=0.8\textwidth]{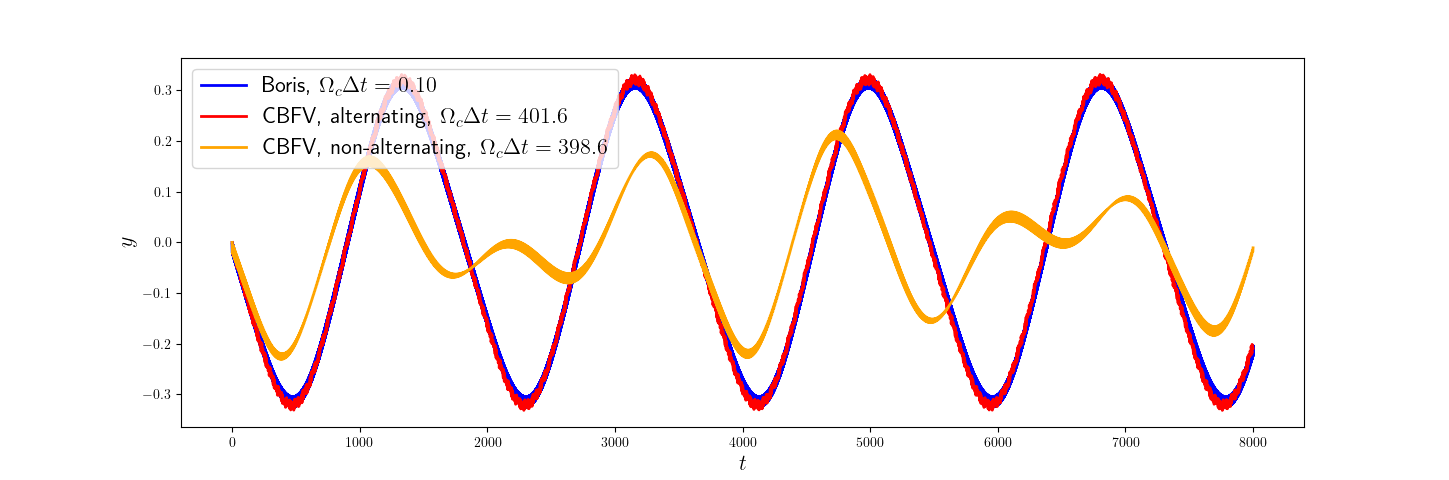}
	\caption{Depiction of orbit  and time-histories for alternating-time-step test problem.  For comparable average time-steps, the alternating-step method reproduces the correct closed, elliptic orbit, while the non-alternating approach does not.} \label{fig:AltTest}
\end{figure}

\subsection{Sinusoidal Variation in Perpendicular Plane with $\nabla B$ drift}
For this test, we set 
\begin{equation} \label{eq:2Dsetup}
	\bB = 100 (1 - x/20) \widehat{\mathbf{z}}, \qquad \bE = \frac{1}{2} \cos (k_x x) \widehat{\mathbf{x}} + \cos (k_y y) \widehat{\mathbf{y}}.
\end{equation}
While several values of $k_x, k_y$ were tested, we report here results for $k_x \rho = \pi/4$, $k_y \rho = \pi/6$, where $\rho$ here is the initial gyroradius of $1/100$.  These wavenumbers are chosen because they lead to a complex orbit arising from competition between the $\nabla B$ drift (in the positive $y$-direction) and $\bE \times \bB$ drift (in the negative $y$-direction), leading to a more challenging test case.  

For this test, we also compare against BFV, CBFV and CN.  Results appear in Figure \ref{fig:2Dsine}.  As expected, CN completely fails to see the $\nabla B$-drift.  BFV and CBFV see the $\nabla B$-drift, but incorrectly predict the confinement in the $x$-direction.  Only CBFV+FLR gives orbits that match Boris.  Lastly, in Figure \ref{fig:2Dsine_energy}, we display energy conservation properties of each scheme.  As discussed above, only $O(\Delta t^2)$ energy conservation is expected in the single-particle context, so it is unsurprising that each scheme has larger energy errors than Boris purely as a result of the much larger time-steps.  It is encouraging, however, to observe that the new modifications presented here do not degrade energy conservation in the single-particle context relative to previous schemes.  

\begin{figure}[h!]
	\centering
	\includegraphics[width=0.6\textwidth]{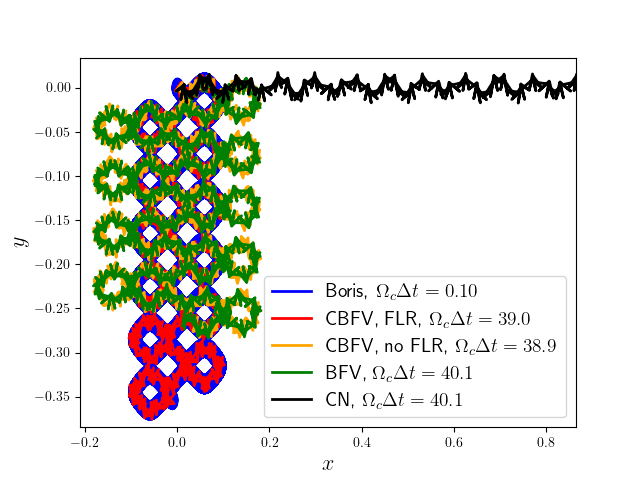}
    \includegraphics[width=0.49\textwidth]{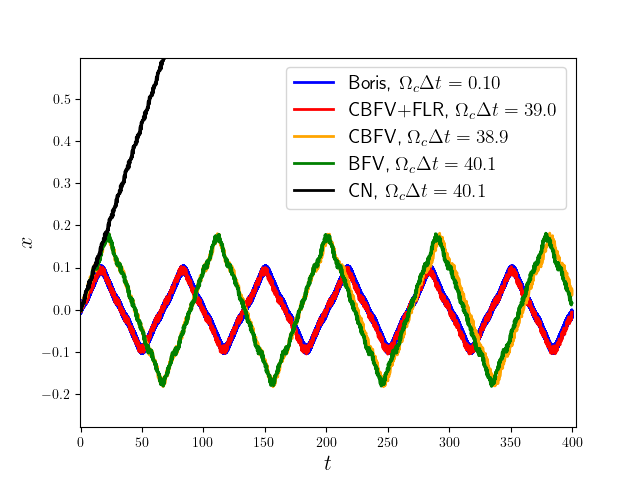}
    \includegraphics[width=0.49\textwidth]{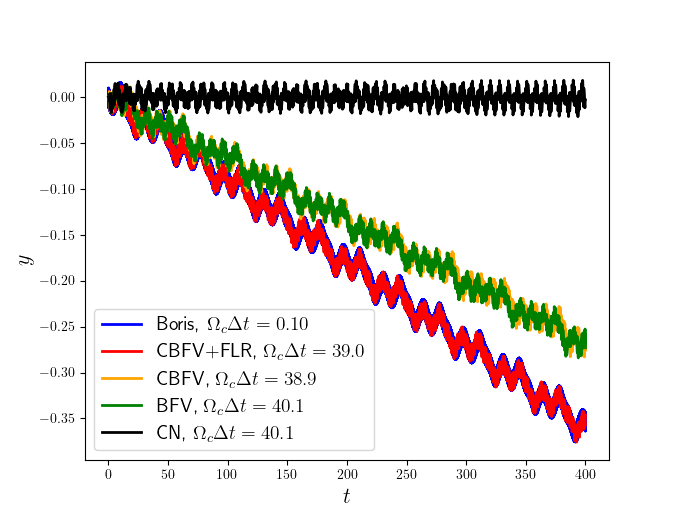}
	\caption{Depiction of orbit for 2-dimensional problem with sinusoidal, gyro-scale variation in electric field.  The FLR corrections presented here are necessary to reproduce the fully-resolved orbit.}
	\label{fig:2Dsine}
\end{figure}

\begin{figure}[h]
	\centering
	\includegraphics[width=0.6\textwidth]{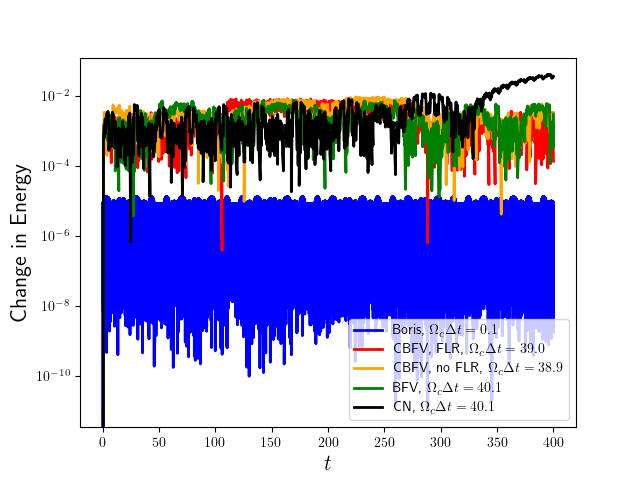}
	\caption{Fractional energy change as a function of time for two-dimensional problem with sinusoidal, gyro-scale variation in electric field.  Exact conservation is not expected here, but the new scheme does not feature secular energy growth and conserves energy no worse than previous schemes.}
	\label{fig:2Dsine_energy}
\end{figure}

\subsection{Tokamak Geometry}
As a final test, we consider a tokamak-like geometry with fixed electric and magnetic fields.  The setup is identical to that in \cite{ricketson2019energy} except for the addition of an electric field with gyration-scale fluctuations, but we reproduce the description here for completeness.  The poloidal magnetic field is based on a simple analytic solution $\psi$ of the Grad-Shafranov equation \cite{cerfon2010one}:
\begin{equation}
	\psi(r,z) = \frac{C}{8} r^4 + d_1 + d_2 r^2 + d_3 \left( r^4 - 4r^2 z^2 \right),
\end{equation}
where $C, d_i$ are constants and $(r, \theta, z)$ is the standard polar coordinate system.  

The flux function $\psi$ specifies the poloidal field by $\bB_p = \nabla \psi \times \mathbf{e}_\theta / r$.  We let $C = 300$, and we take advantage of the fact, shown in \cite{pataki2013fast}, that there is a one-to-one linear map between the $d_i$ and the tokamak shape parameters $\varepsilon$ (inverse aspect ratio), $\kappa_T$ (elongation), and $\delta$ (triangularity).  Thus, the poloidal field is uniquely specified by choosing the shape parameters of the International Thermonuclear Energy Reactor (ITER): $\varepsilon = 0.32$, $\kappa_T = 1.7$, $\delta = 0.33$.  The toroidal field is chosen to be $\bB_{tor} = 800 \mathbf{e}_\theta / r$.  

We additionally impose an electrostatic potential that is a flux function (i.e. expressible as a function of $\psi$) such that $\bE\cdot\bB=0$ according to the ideal MHD Ohm's law, $\mathbf{E} + \bv \times \bB = 0$.  That flux function is chosen to be $\phi = \sin (K \psi) / 2 K$.  $K$ is chosen to target a specific characteristic value of $k_\perp$ in the electric field by computing $K = k_\perp / \| \nabla \psi (\bx^0) \|$.

The particle is initialized at the Cartesian location $\mathbf{x}^0 = 1.2\widehat{\mathbf{x}}$, 
at which point the magnitude of the magnetic field is $B \approx 667.7$.  This knowledge allows us to compute the gyroradius $\rho$ at the initial particle location, and to choose $\kappa$ so that $k_\perp \rho = 1.5$ at the initial particle location.  The value of $k_\perp \rho$ varies in space, but remains $O(1)$ along the entire particle trajectory.  The initial velocity is given in Cartesian coordinates by $\bv^0 = (1, v_y, 0)^T$.  Note that as the field is predominantly toroidal (i.e., along the $y$-axis) at the initial location, $v_y$ may be taken as a reasonable proxy for the particle's initial parallel velocity.  We choose $v_y = 0.6$ because it results in an orbit near the trapped-passing boundary.  Simulations are run to a final time of $T=400$, which corresponds to several traversals of the banana orbit.    

\begin{figure}[h]
	\centering
	\includegraphics[width=0.49\textwidth]{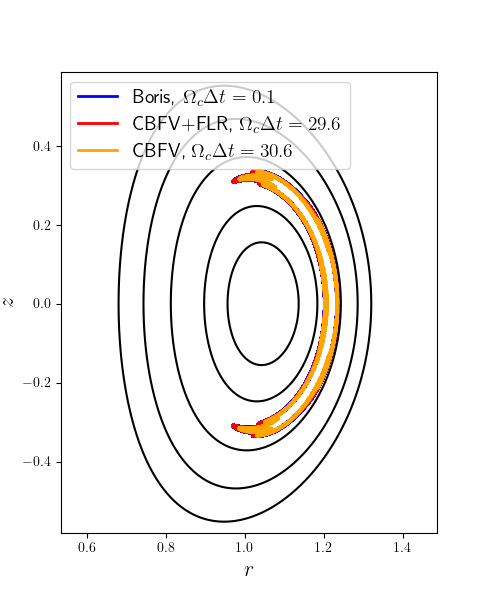}
	\includegraphics[width=0.49\textwidth]{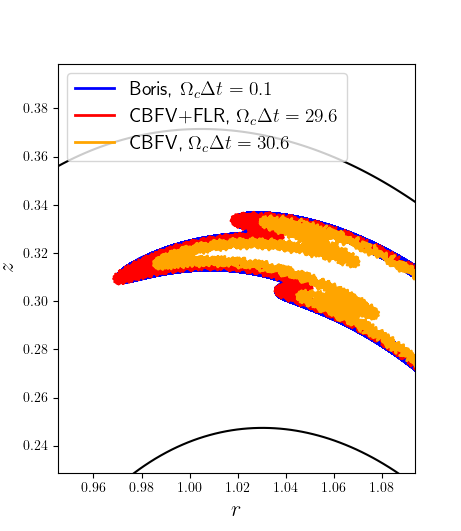}
	\caption{Projection into the poloidal plane of particle orbits for the tokamak equilibrium test case.  Different curves are trajectories resulting from well-resolved Boris integration (blue), the new scheme with FLR corrections (red), and the previous scheme without FLR corrections (orange).  Single lines represent flux surfaces of the Solov'ev equilibrium magnetic field.  Full orbit appears at left, with a zoom-in of the upper turning point on the right.  The modified structure of the orbit due to fine-scale variation in the electric field is clearly visible.  In the poloidal projection of the orbit, the FLR corrections seem to afford only mild improvements in accuracy in capturing the turning points, although this betrays catastrophic errors in their timing, as demonstrated in the next figure.} 
	\label{fig:tokamak_vpar6}
\end{figure}

To make plots more readable, we only show results from Boris, CBFV, and CBFV+FLR for this test (although we have confirmed that the other integrators fail as catastrophically as in the previous test).  Only slight improvements due to the FLR corrections are visible from a poloidal projection of the particle orbit, shown in Figure \ref{fig:tokamak_vpar6}.  However, without FLR corrections, CBFV results in the particle traversing this orbit at the wrong rate.  This is shown in Figure \ref{fig:rzt_vx1}, in which time-histories of $r$ and $z$ are reported.  The rate of traversal is dramatically improved by the FLR corrections, resulting in only a slight deviation from fully-resolved Boris over several banana orbits.  

\begin{figure}[h]
	\centering
	\includegraphics[width=0.9\textwidth]{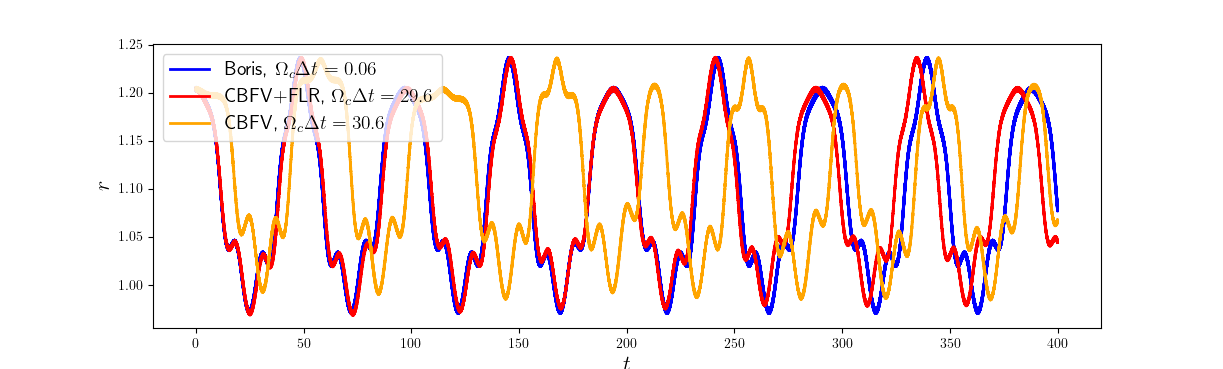}
	\includegraphics[width=0.9\textwidth]{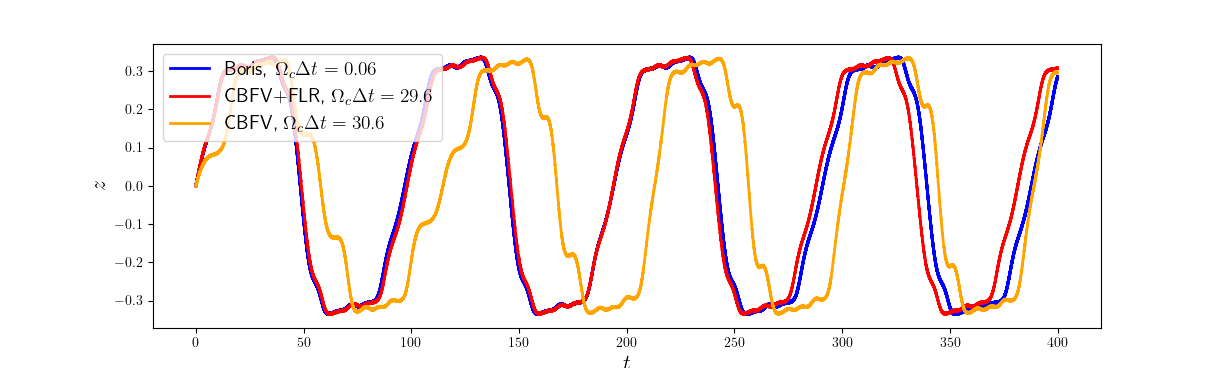}
	\caption{Time histories of $r$ and $z$ for the tokamak equilibrium test case.  Note the dramatically improved temporal accuracy of the new integrator with FLR corrections.}
	\label{fig:rzt_vx1}
\end{figure}

%\begin{figure}[h]
%	\centering
%	\includegraphics[width=0.6\textwidth]{Figures/Energy_conservation_tokamak_new.png}
	%\caption{Energy change as a function of time for well-resolved Boris (blue), the new scheme with FLR corrections (red), and the previous scheme without FLR corrections (orange).  Note that none of the schemes features secular growth in energy error, in spite of the occasional discontinuities in electrostatic potential encountered by the particle. \lc{This point has already been made. We should consider removing. If not, pull legend down, and increase font size.}} 
	%\label{fig:energy_vpar0p4}
%\end{figure}

%Finally, we report the magnitude of change in energy for each scheme in Figure \ref{fig:energy_vpar0p4}.  Note that, as predicted, exact energy conservation is not recovered here due to the spatially varying electric field in the single particle context.  However, no secular growth or decay in energy is observed, and indeed energy errors are comparable to Boris in spite of using a time-step 300-times larger.  

%%%%%%%%%%%%%%%%%%%%%%%%%%%%%%%%%%%%%%%%%%%%%%%%%%%

\section{Conclusions}
In this work, we have built on the asymptotic-preserving and energy-conserving integrator of \cite{ricketson2019energy} with two modifications that allow it to capture FLR effects even when stepping over the gyration time-scale.  Those modifications are (a) introducing an explicit and adaptive gyro-average computation of $\bE$, and (b) a strategy for alternating large and small time-steps so that the particle samples evenly spaced gyrophases over successive time-steps.  We have shown that with the first modification, the scheme still reproduces the proper gyro-orbit while the time-step constraints on the scheme are considerably relaxed due to the second.  We have also shown that with an analogous modification to current deposition, the scheme remains compatible with the exact energy conservation enjoyed by existing implicit PIC schemes.  The scheme was tested on single-particle motion with several specified electromagnetic fields, demonstrating the ability of the scheme to accurately predict orbits even when the electric field varies on scales comparable to the gyroradius.  

There are, of course, opportunities for further development.  Of obvious interest is the implementation of this scheme in implicit PIC codes, and its subsequent use in simulating physically relevant phenomena.  Such an implementation will certainly benefit from a more robust and general adaptive time-stepping scheme that can smoothly and accurately handle transitions between magnetized and unmagnetized regimes in the presence of the alternating time-step approach.  Additionally, local charge conservation in the presence of variable $n_g$ requires further investigation.  Each of these is an active area of current research.

%%%%%%%%%%%%%%%%%%%%%%%%%%%%%%%%%%%%%%%%%%%%%%%%%%%

\section*{Acknowledgements}
The authors wish to acknowledge private communication with Matthew Miecnikowski, whose input directly led to the idea of alternating large and small time-steps.  Additional discussions with C-S Chang, Guangye Chen, Joshua Burby, and Davide Curreli were valuable in the development of this work.  

%%%%%%%%%%%%%%%%%%%%%%%%%%%%%%%%%%%%%%%%
%%%%%%%%%%%%%%%%%%%%%%%%%%%%%%%%%%%%%%%%
%%%%%%%%%%%%%%%%%%%%%%%%%%%%%%%%%%%%%%%%

\appendix
\section{Proof of Lemma 1}
It suffices to show that any four consecutive particle locations lie on a single circle of radius $\rho = u/\Oc$.  This extends to \textit{all} points by induction as follows.  Suppose $\{ \bx^{k}, \bx^{k+1}, \bx^{k+2}, \bx^{k+3} \}$ lie on a circle of radius $\rho$, as do $\{ \bx^{k+1}, \bx^{k+2}, \bx^{k+3}, \bx^{k+4} \}$.  Then, $\{ \bx^{k+1}, \bx^{k+2}, \bx^{k+3} \}$ lie on a shared circle of radius $\rho$.  Three non-colinear points define a unique circle.  Thus, all five points $\{ \bx^{k}, \bx^{k+1}, \bx^{k+2}, \bx^{k+3}, \bx^{k+4} \}$ lie on that same unique circle of radius $\rho$, and the induction carries.  

Without loss of generality, we proceed by showing that the first four points $\{ \bx^k :\, k=0,1,2,3 \}$ lie on a single circle of radius $\rho$.  As already noted, the velocity update in the perpendicular direction may be written 
\begin{equation}
	\bu^{n+1} = R_{\theta_n} \bu^n, 
\end{equation}
where $\theta_n$ is defined by
\begin{equation} \label{eq:thetandef}
	\cos \theta_n = \frac{1 - \Oc^2 \Dt_n^2 / 4}{1 + \Oc^2 \Dt_n^2 / 4}.
\end{equation}
This already demonstrates that the magnitude of $\bu^n$ is independent of $n$.  We denote this common magnitude by $u$ for the remainder of the proof.  

Moreover, it is a simple algebraic exercise -- see, e.g.\ \cite{ricketson2019energy} -- to show that
\begin{equation} \label{eq:halfstepu}
	\bu^{n+1/2} = \frac{\bu^n + \frac{1}{2} \Oc \Dt_n \bu^n \times \bb}{1 + \Oc^2 \Dt_n^2 / 4}.
\end{equation}
Taking the magnitude of this expression gives
\begin{equation}
	u\nph = \frac{u}{\sqrt{1 + \Oc^2 \Dt_n^2 / 4}}.
\end{equation}

One may thus write an expression for the particle displacement magnitude $\Delta x^n = \left\| \Delta \bx^n \right\| = \left\| \bx^{n+1} - \bx^n \right\|$ at each time-step:
\begin{equation}
	\Delta x^n = u\nph \Delta t_n = \frac{u \Dt_n}{\sqrt{1 + \Oc^2 \Dt_n^2 / 4}}.
\end{equation}
Using the trigonometric identity $\sin \left( \theta_n / 2 \right) = \sqrt{(1 - \cos \theta_n)/2}$ and the definition of $\theta_n$ in \eqref{eq:thetandef}, we find
\begin{equation} \label{eq:dxmag}
	\Delta x^n = 2 \rho \sin \left( \frac{\theta_n}{2} \right).
\end{equation}

Further, as discussed in the main text, the angle between the displacements $\Delta \bx^n$ and $\Delta \bx^{n+1}$ is necessarily $(\theta_n + \theta_{n+1})/2$.  For brevity,  we define $\theta_{n+1/2} = (\theta_n + \theta_{n+1})/2$.  The law of sines thus gives the following expression for the radius $r$ of the unique circle containing $\{ \bx^0, \bx^1, \bx^2 \}$:
\begin{equation} \label{eq:rexpr}
	r = \frac{1}{2} \frac{ \left\| \bx^2 - \bx^0 \right\| }{\sin \left( \pi - \frac{\theta_0 + \theta_1}{2} \right)} = \frac{1}{2} \frac{ \left\| \bx^2 - \bx^0 \right\| }{\sin \theta_{1/2}}.
\end{equation}

The numerator may be simplified by applying the law of cosines:
\begin{equation}
\begin{split}
	\left\| \bx^2 - \bx^0 \right\| &= \left[ \left(\Delta x^0 \right)^2 + \left(\Delta x^1 \right)^2 + 2 \Delta x^0 \Delta x^1 \cos \theta_{1/2} \right]^{1/2} \\
	&= 2\rho \left[ \sin^2 \frac{\theta_0}{2} + \sin^2 \frac{\theta_1}{2} + 2 \sin \frac{\theta_0}{2} \sin \frac{\theta_1}{2} \cos \theta_{1/2} \right]^{1/2}.
\end{split}
\end{equation}
This may be further simplified by applying an angle sum identity to $\cos \theta_{1/2}$:
\begin{equation} \label{eq:bigdx}
\begin{split}
	\left\| \bx^{2} - \bx^{0} \right\| &= 2 \rho \left[ \sin^2 \frac{\theta_0}{2} + \sin^2 \frac{\theta_1}{2} + 2 \sin \frac{\theta_0}{2} \sin \frac{\theta_1}{2} \left( \cos \frac{\theta_0}{2} \cos \frac{\theta_1}{2} - \sin \frac{\theta_0}{2} \sin \frac{\theta_1}{2} \right) \right]^{1/2} \\
	&= 2 \rho \left[ \sin^2 \frac{\theta_0}{2} \cos^2 \frac{\theta_1}{2} + \cos^2 \frac{\theta_0}{2} \sin^2 \frac{\theta_1}{2} + 2 \sin \frac{\theta_0}{2} \sin \frac{\theta_1}{2} \cos \frac{\theta_0}{2} \cos \frac{\theta_1}{2} \right]^{1/2} \\
	&= 2 \rho \left[ \left( \sin \frac{\theta_0}{2} \cos \frac{\theta_1}{2} + \cos \frac{\theta_0}{2} \sin \frac{\theta_1}{2} \right)^2 \right]^{1/2} \\
	&= 2 \rho \left[ \sin^2 \theta_{1/2} \right]^{1/2} \\
	&= 2 \rho \sin \theta_{1/2}.
\end{split}
\end{equation}
Substituting this into \eqref{eq:rexpr}, we have $r=\rho$.  Thus, $\{ \bx^0, \bx^1, \bx^2 \}$ lie on a circle of radius $\rho$.  As we have not used any special properties of the first three points, it also follows that $\{ \bx^1, \bx^2, \bx^3 \}$ lie on a circle of radius $\rho$.  

Now, $\{ \bx^0, \bx^1, \bx^2 \}$ and $\{ \bx^1, \bx^2, \bx^3 \}$ each lie on a circle of radius $\rho$ and they share two points.  There are thus at most two circles of radius $\rho$ on which these four points can lie.  This is depicted graphically below in Figure \ref{fig:twoCircles}.
\begin{figure}[h]
	\centering
	\includegraphics[width=0.6\textwidth]{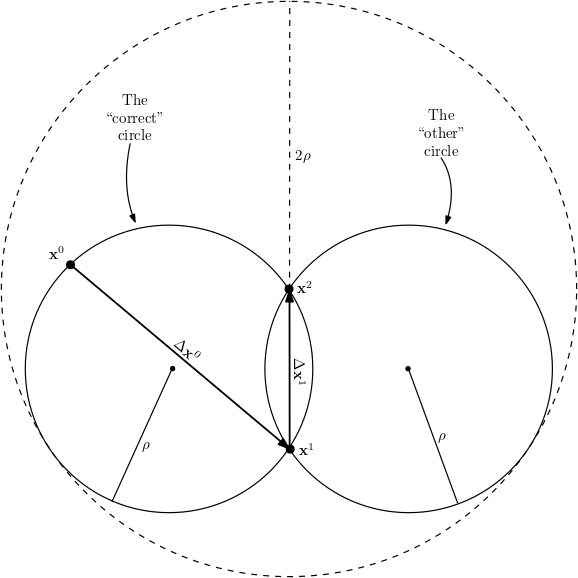}
	\caption{Illustration of the two possible (solid) circles on which $\bx^3$ may lie.  In the limit $\Oc \Dt \rightarrow \infty$, $\bx^3$ must also lie on the dashed circle, so it must lie on one of the two points of tangency.  In the text, it is shown that only the ``correct" tangency point (the leftward one) is possible.} \label{fig:twoCircles}
\end{figure}

Our aim is to show that $\bx^3$ lies on the same circle as $\{ \bx^0, \bx^1, \bx^2 \}$.  To this end, consider first the limit $\Oc \Dt_2 \rightarrow \infty$.  Clearly, according to \eqref{eq:thetandef} and \eqref{eq:dxmag}, we have $\Delta x^2 \rightarrow 2 \rho$ in this limit.  Thus, in this limit, $\bx^3$ must lie on one of the two circles of radius $\rho$ appearing in Figure \ref{fig:twoCircles}, \textit{and} on a circle of radius $2\rho$ centered at $\bx^2$ -- the dashed circle in Figure \ref{fig:twoCircles}.  $\bx^3$ thus lies on one of the points of tangency between the circles of radius $\rho$ and that of $2\rho$.  Note now that, in order to lie on the point of tangency located on the circle of radius $\rho$ that \textit{does not} contain $\bx^0$, it must be the case that the particle velocity rotates in the \textit{opposite} direction in moving from $t_2$ to $t_3$ as it has previously (or to rotate by an angle larger than $\pi$, which is also impossible).  This would require $\bB$ to change direction, which is ruled out by assumption.  

Thus, at least in the $\Oc \Dt_2 \rightarrow \infty$ limit, $\bx^3$ lies on the circle of radius $\rho$ defined by $\{ \bx^0, \bx^1, \bx^2 \}$.   We again appeal to \eqref{eq:thetandef} and \eqref{eq:dxmag} to see that the particle displacement is a smooth function of $\Dt_2$.  For $\bx^3$ to jump from one of the two radius-$\rho$ circles in Figure \ref{fig:twoCircles} to the other as $\Dt$ varies would require a discontinuity at least in the derivative $d \Delta \bx^2 / d \Dt_2$.  As no such discontinuity exists, we conclude that $\{ \bx^0, \bx^1, \bx^2, \bx^3 \}$ lie on a single circle of radius $\rho$ independent of $\Dt_n$.  

As we have nowhere exploited the fact that $\bx^0$ is the initial particle location, the same conclusion holds for $\{ \bx^n, \bx^{n+1}, \bx^{n+2}, \bx^{n+3} \}$.  Thus, the induction carries and the proof is complete.  \qed

%%%%%%%%%%%%%%%%%%%%%%%%%%%%%%%%%%%%%%%%

\section{Proof of Corollary 1}
The first limit (i) is clear.  In the small step limit, the orbit converges to the analytic orbit.  The first two time-steps specify $\{ \bx^0, \bx^1, \bx^2 \}$, which in turn specify a unique circle of radius $\rho$ on which all future points lie.  Thus, if $\{ \bx^0, \bx^1, \bx^2 \}$ converge to their true values, then the numerical orbit they specify coincides with the true orbit.  

In considering the second limit (ii), we first note that the true gyro-center is located at 
\begin{equation}
	\bx_{gc} = \bx_0 + \Oc^{-1} \bu^0 \times \bb = \bx^0 + \rho \widehat{\bu}^0 \times \bb, 
\end{equation}
where $\widehat{\bu}^0 = \bu^0/u^0$.  Next, we return to \eqref{eq:halfstepu} and note that
\begin{equation}
	\lim_{\Oc \Dt_0 \rightarrow \infty} \bu^{1/2} = 2 \frac{ \bu^0 \times \bb}{\Oc \Dt_0} = 2 \frac{\rho}{\Dt_0} \widehat{\bu}^0 \times \bb.
\end{equation}
Thus, the next particle location is 
\begin{equation}
	\lim_{\Oc \Dt_0 \rightarrow \infty} \bx^1 = \bx^0 + \Dt_0 \bu^{1/2} = \bx^0 + 2\rho \widehat{\bu}^0 \times \bb.  
\end{equation}

We see that, in this limit, $\bx^1$ and $\bx^0$ are connected by a diameter of the \textit{true} gyro-orbit.  Moreover, there is a unique circle of radius $\rho$ that contains this diameter.  Thus, all future particle locations share this circle, and thus lie on the true gyro-orbit.  \qed

%%%%%%%%%%%%%%%%%%%%%%%%%%%%%%%%%%%%%%%%

\section{Derivation of Time-Step Restrictions}
As noted in the main text, we consider only the case $u^{n+1/2} > v_E$.  The expression for $\bF_{eff}$ simplifies considerably in this case, allowing us to write
\begin{equation}
	\Delta \bv_{\nabla B} = \frac{\bb}{\Oc} \times \frac{\tilde{\mu}}{m} \left\{ \frac{v_\parallel}{v_\perp} \widehat{\bv}_\perp (\nabla B)_\parallel + ( \bI - 2 \widehat{\bv}_\perp \widehat{\bv}_\perp ) \cdot (\nabla B)_\perp \right\},
\end{equation}
where
\begin{equation}
	\tilde{\mu} = \frac{\Oc^2 \Dt^2 / 4}{1 + \Oc^2 \Dt^4/4} \mu
\end{equation}
for the large time-step, and analogous for the small time-step.  
We break this anomalous velocity into two components to be analyzed separately:
\begin{equation} \label{eq:v_anom_parts}
	\Delta \bv_{\nabla B}^1 = \frac{\bb}{\Oc} \times \frac{\tilde{\mu}}{m} \frac{v_\parallel}{v_\perp} (\nabla B)_\parallel \widehat{\bv}_\perp, \qquad \Delta \bv_{\nabla B}^2 = \frac{\bb}{\Oc} \times \frac{\tilde{\mu}}{m} \left( \bI - 2 \widehat{\bv}_\perp \widehat{\bv}_\perp \right) \cdot (\nabla B)_\perp.
\end{equation}
As before, quantities without superscripts are assumed to be evaluated at the half-step.

We will denote the anomalous displacements arising from these anomalous velocities by $\mathbf{a}^1$ and $\mathbf{a}^2$, defined respectively by
\begin{equation}
	\mathbf{a}^{i,n+1} = \mathbf{a}^{i,n} + \Delta t_n \Delta \mathbf{v}_{\nabla B}^{i,n+1/2}
\end{equation}
for $i=1,2$.  In the uniform time-step case of \cite{ricketson2019energy}, it was noted that these respective anomalous displacements have constant magnitude and rotate by a fixed angle at every time-step - $\theta$ for $\mathbf{a}^1$, and $2\theta$ for $\mathbf{a}^2$.  A simple geometric argument then gave formulae for the radii of the circles traversed by $\mathbf{a}^1$ and $\mathbf{a}^2$.  Bounding these radii in terms of physical parameters then led immediately to time-step restrictions.  

We follow the same general approach here, but the geometric arguments differ because the displacement magnitude is no longer constant due to the variation in time-step size.  As such, we will use similar arguments to those used in Appendix A.  To that end, note that the rotation angle of each displacement is still constant - $\psi$ for $\mathbf{a}^1$, and $2\psi$ of $\mathbf{a}^2$.  Here $\psi = \pi ( 1 - 1/m )$ is as defined in the main text.  We will denote displacements in the large and small time-steps by $\Delta \mathbf{a}^i$ and $\delta \mathbf{a}^i$, respectively.  

We begin with the second anomalous motion $\mathbf{a}^2$.  By taking magnitudes, we have
\begin{equation}
	\Delta a^2 = \frac{\Oc^2 \Dt^2 / 4}{1 + \Oc^2 \Dt^2 / 4} \frac{\mu}{m \Oc} \left\| \left( \nabla B \right)_\perp \right\| \Dt, \qquad \delta a^2 = \frac{\Oc^2 \delta t^2 / 4}{1 + \Oc^2 \delta t^2 / 4} \frac{\mu}{m \Oc} \left\| \left( \nabla B \right)_\perp \right\| \delta t.
\end{equation}
Some algebra yields
\begin{equation}
	\Delta a^2 = \left( 1 - \cos \theta \right) \frac{\delta_\perp}{4} \Oc \Dt \rho, \qquad \delta a^2 = \left( 1 - \cos \delta \theta \right) \frac{\delta_\perp}{4} \Oc \delta t \rho,
\end{equation}
where we have used
\begin{equation}
	\frac{\Oc^2 \Dt^2 / 4}{1 + \Oc^2 \Dt^2 / 4} = \frac{1}{2} \left[ 1 - \frac{1 - \Oc^2 \Dt^2 / 4}{1 + \Oc^2 \Dt^2 / 4}\right] = \frac{1 - \cos \theta}{2}
\end{equation}
along with the analogous reasoning for the small time-step.  As in Appendix A, we now use the law of sines to write an expression for the radius $r$ of a circle containing three consecutive $\mathbf{a}^2$ points:
\begin{equation}
	2r = \frac{\left[ \left( \Delta a^2 \right)^2 + \left( \delta a^2 \right)^2 + 2 \Delta a^2 \delta a^2 \cos \left( 2\psi \right) \right]^{1/2}}{\lvert \sin \left( 2\psi \right) \rvert}.
\end{equation}

Evidently, $r$ does not depend on time-step number, so we can use identical logic to Appendix A to conclude that \textit{all} $\mathbf{a}^2$ points lie on a single circle with this radius.  By using trivial bounds on geometric functions, we can find a simple upper bound on the radius of this circle:
\begin{equation}
	r \leq \frac{1}{2} \frac{\rho \delta_\perp}{\sin \left(\frac{2\pi}{m} \right)} \Oc \left( \frac{\Dt + \delta t}{2} \right).
\end{equation}
where we have substituted in the definition of $\psi$ in terms of the number of desired gyrophase samples $m$.  If the length scale we wish to resolve is $L_{res}$, a sufficient condition to make sure the anomalous displacement is smaller than that scale is thus
\begin{equation} \label{eq:restric1}
	\Oc \left( \frac{\Dt + \delta t}{2} \right) < \left( \frac{L_{res}}{\rho} \right) \frac{\sin \left( \frac{2\pi}{m} \right)}{\delta_\perp}.
\end{equation}

We now return to $\mathbf{a}^1$, which is analyzed in the same way.  Taking magnitudes, we have
\begin{equation}
	\Delta a^1 = \frac{1}{2} \frac{\Oc^3 \Dt^3 / 4}{\sqrt{1 + \Oc^2 \Dt^2 / 4}} \delta_\parallel \rho, \qquad \delta a^1 = \frac{1}{2} \frac{\Oc^3 \delta t^3 / 4}{\sqrt{1 + \Oc^2 \delta t^2 / 4}} \delta_\parallel \rho.
\end{equation}
The root in the denominator arises from the evaluation of $v_\perp$ at the half-step, and the assumption that perpendicular motion is dominated by gyration, giving $v_\perp \approx u^n / \sqrt{1 + \Oc^2 \Dt^2/4}$ and similarly for the small step.  

Anticipating that we'll resort to inequalities as before, the expressions for the displacements can be reduced to upper bounds:
\begin{equation}
	\Delta x^1 \leq \frac{1}{4} \Oc^2 \Dt^2 \delta_\parallel \rho, \qquad \delta x^1 \leq \frac{1}{4} \Oc^2 \delta t^2 \delta_\parallel \rho.
\end{equation}
Again substituting into the law of sines with the fact that the rotation angle is now $\psi$, we have
\begin{equation}
	r \leq \frac{1}{2} \frac{\rho \delta_\parallel}{\sin \left( \frac{\pi}{m} \right)} \left[ \Oc \left( \frac{\Dt + \delta t}{2} \right) \right]^2.
\end{equation}
We again requre that $2r < L_{res}$.  A sufficient condition for satisfying this is
\begin{equation} \label{eq:restric2}
	\Oc \left( \frac{\Dt + \delta t}{2} \right) < \left[ \left( \frac{L_{res}}{\rho} \right) \frac{\sin \left( \frac{\pi}{m} \right)}{\delta_\parallel} \right]^{1/2}.
\end{equation}
The time-step restriction \eqref{eq:spacerestric} in the main text simply enforces both \eqref{eq:restric1} and \eqref{eq:restric2}.  

%%%%%%%%%%%%%%%%%%%%%%%%%%%%%%%%%%%%%%%%
\bibliographystyle{plain}
\bibliography{APbib}

\end{document}